\documentclass{amsart}

\usepackage[table]{xcolor}
\usepackage[normalem]{ulem}  
\usepackage{environ}

\definecolor{addedcolor}
  {rgb}{0.0, 0.7, 0.7}
\newif\ifdraftmode

\draftmodefalse

\newcommand{\addedinline}[1]{%
  \ifdraftmode
    \textcolor{addedcolor}{#1}%
  \else
    #1%
  \fi
}

\definecolor{deletedcolor}{rgb}{0.8, 0.0, 0.0} 

\NewEnviron{deleted}{%
  \ifdraftmode
    \color{deletedcolor}\BODY
  \else
  \fi
}

\allowdisplaybreaks

\usepackage{caption}
\usepackage[leftcaption]{sidecap}

\usepackage[inline]
  {enumitem}
\setlist[enumerate,1]{label=\bf (\roman*)}
\setlist{
  leftmargin=.65cm,
  before=\leavevmode
}

\usepackage[
  backend=biber,
  style=alphabetic,
  giveninits=true, 
  backref=true
]{biblatex}
\usepackage{lmodern}
\usepackage{anyfontsize}
    
\usepackage{multirow}
\usepackage{hhline}

\usepackage{tabularray}
\UseTblrLibrary{booktabs}

\usepackage{adjustbox}
\usepackage{tcolorbox}

\newtcolorbox{standout}{
  colback=gray!15,
  boxrule=0pt,
  left=.3cm,
  right=.3cm,
  top=.18cm,
  bottom=.18cm,
  boxsep=0pt
}

\usepackage{mathtools}
\usepackage{amssymb}
\usepackage{amsthm}
\usepackage{xfrac}
\usepackage{stmaryrd} 
\usepackage{scalerel} 
\usepackage{stackengine} 

\newcommand{\bracket}[3]{%
  \stretchleftright
    {#1}
    {\ensurestackMath{\addstackgap[1pt]{#2}}}
    {#3}%
}
\newcommand{\bracketmid}[4]{%
  \stretchleftright{#1}{%
    \ensurestackMath{%
      \addstackgap[2pt]{#2}%
      \,\stretchrel*{|}{\addstackgap[2pt]{#2#3}}\,%
      \addstackgap[2pt]{#3}%
    }%
  }{#4}%
}

\theoremstyle{plain}
\newtheorem{theorem}{Theorem}[section]
\newtheorem{lemma}[theorem]{Lemma}
\newtheorem{proposition}[theorem]{Proposition}

\theoremstyle{definition}
\newtheorem{definition}[theorem]{Definition}
\newtheorem{example}[theorem]{Example}

\theoremstyle{remark}
\newtheorem{remark}[theorem]{Remark}
\usepackage[
  colorlinks=true,
  linkcolor=darkgreen,
  citecolor=darkgreen,
  urlcolor=darkgreen
]{hyperref}
\usepackage{xurl} 

\usepackage{cleveref}
\crefname{equation}{}{}
\crefname{section}{\S}{\S\S}
\crefname{subsection}{\S}{\S\S}
\crefname{subsubsection}{\S}{\S\S}
\crefname{definition}{Def.}{Defs.}
\crefname{theorem}{Thm.}{Thms.}
\crefname{corollary}{Cor.}{Cors.}
\crefname{lemma}{Lem.}{Lems.}
\crefname{proposition}{Prop.}{Props.}
\crefname{remark}{Rem.}{Rems.}
\crefname{notation}{Ntn.}{Ntns.}
\crefname{fact}{Fact}{Fact}
\crefname{example}{Ex.}{Exs.}
\crefname{figure}{Fig.}{Figs.}
\crefname{table}{Tab.}{Tabs.}
\crefname{footnote}{ftn.}{ftns.}
\Crefname{footnote}{Ftn.}{Ftns.}

\usepackage{tikz}
\usetikzlibrary{
  cd,
  calc,
  arrows.meta,
  backgrounds,
  decorations,
  decorations.pathmorphing, 
}

\tikzcdset{
  arrow style=tikz, 
  diagrams={
    trim left=
    {([xshift=5pt]current bounding box.west)},
    trim right=
    {([xshift=-5pt]current bounding box.east)},
    baseline=
    {([yshift=-2pt]current bounding box.center)},
    >={
      Computer Modern Rightarrow[
        length=4pt, width=4pt
      ]
    },
    hook/.style={
      {Hooks[right, length=2pt, width=8pt]}->
    },
    hook'/.style={
      {Hooks[left, length=2pt, width=8pt]}->
    },
    shorten=0pt,
  }
}

\definecolor{darkblue}{rgb}{0.05,0.25,0.65}
\definecolor{darkgreen}{RGB}{20,140,10}
\definecolor{lightgray}{rgb}{0.9,0.9,0.9}
\definecolor{darkorange}{RGB}{200,100,5}
\definecolor{darkyellow}{rgb}{.91,.91,0}
\definecolor{lightolive}{RGB}{225, 220, 185}

\let\originalsslash\sslash
\renewcommand{\sslash}{\mathord{\originalsslash}}

\makeatletter
\newcommand{\cpt}{\mathpalette\cpt@inner\relax}
\newcommand{\cpt@inner}[2]{%
  \scalebox{0.5}[0.9]{$#1\cup$}
  #1\{\infty\}
}
\makeatother

\newcommand{\grayunderbrace}[2]{\mathcolor{gray}{\underbrace{\mathcolor{black}{#1}}}_{\mathcolor{gray}{#2}}}

\newcommand{\grayoverbrace}[2]{\mathcolor{gray}{\overbrace{\mathcolor{black}{#1}}}^{\mathcolor{gray}{#2}}}

\tikzset{
  snake left/.style={
    rounded corners,
    to path={
      let \p1 = (\tikztostart.east),
          \p2 = (\tikztotarget.west),
          \p3 = ($(\p1)!0.5!(\p2)$),
          \n1 = {8pt} 
      in
      (\p1)
      -- (\x1 + \n1, \y1)
      -- (\x1 + \n1, \y3)
      -- (\x2 - \n1, \y3) \tikztonodes
      -- (\x2 - \n1, \y2)
      -- (\p2)
    }
  }
}

\tikzset{
  uphordown/.style={
    rounded corners,
    to path={
      let \p1 = (\tikztostart.north),
          \p2 = (\tikztotarget.north),
          \n1 = {max(\y1,\y2) + 8pt}
      in
      (\p1)
      -- (\x1, \n1)
      -- (\x2, \n1) \tikztonodes 
      -- (\p2)
    }
  }
}

\tikzset{
  downhorup/.style={
    rounded corners,
    to path={
      let \p1 = (\tikztostart.south),
          \p2 = (\tikztotarget.south),
          \n1 = {min(\y1,\y2) - 8pt}
      in
      (\p1)
      -- (\x1, \n1)
      -- (\x2, \n1) \tikztonodes 
      -- (\p2)
    }
  }
}

\tikzset{
  rightvertleft/.style={
    rounded corners,
    to path={
      let \p1 = (\tikztostart.east),
          \p2 = (\tikztotarget.east),
          \n1 = {max(\x1,\x2) + 8pt}
      in
      (\p1)
      -- (\n1, \y1)
      -- (\n1, \y2) \tikztonodes 
      -- (\p2)
    }
  }
}

\tikzset{
  leftvertright/.style={
    rounded corners,
    to path={
      let \p1 = (\tikztostart.west),
          \p2 = (\tikztotarget.west),
          \n1 = {min(\x1,\x2) - 8pt}
      in
      (\p1)
      -- (\n1, \y1)
      -- (\n1, \y2) \tikztonodes 
      -- (\p2)
    }
  }
}

\newcommand{\inlinetikzcd}[1]{\begin{tikzcd}[sep=small, ampersand replacement=\&]#1\end{tikzcd}}

\newcommand{\CyclicGroup}[1]{\mathbb{Z}_{/#1}}
\newcommand{\CyclicReals}[1]{\mathbb{R}_{/#1}}

\newcommand{\proofstep}[1]{\scalebox{.7}{#1}}

\newcommand{\defneq}{\equiv}

\newcommand{\HilbertSpace}{%
  \mathcal{H}%
}

\newcommand{\Maps}{\mathrm{Map}}
\newcommand{\MapsComponent}[1]{\mathrm{Map}_{#1}}
\newcommand{\PointedMaps}{\mathrm{Map}^\ast}
\newcommand{\PointedMapsComponent}[1]{\mathrm{Map}^\ast_{#1}}

\newcommand{\periodicity}{2C}

\begin{document}

$\,$
\vspace{-2cm} 
\title
[
  Higher Anyons via Higher Cohomotopy
]
{
  Higher-Dimensional Anyons 
  \\ 
  via Higher Cohomotopy
}

\thanks{\emph{Funding} by Tamkeen UAE under the 
NYU Abu Dhabi Research Institute grant {\tt CG008}.}

\author{Sadok Kallel}
\address{Department of Mathematics and Statistics, American University of Sharjah, UAE}
\curraddr{}
\email{skallel@aus.edu}
\thanks{}

\author{Hisham Sati}
\address{Mathematics Program and Center for Quantum and Topological Systems, New York University Abu Dhabi, UAE}
\curraddr{}
\email{hsati@nyu.edu}
\thanks{}

\author{Urs Schreiber}
\address{Mathematics Program and Center for Quantum and Topological Systems, New York University Abu Dhabi, UAE}
\curraddr{}
\email{us13@nyu.edu}
\thanks{}

\subjclass[2020]{%
Primary: 
55Q55, 
20F18, 
55Q15, 
55Q25, 
81V27; 
Secondary:
81V70, 
55P35
}

\keywords{%
  integer Heisenberg groups;
  Cohomotopy;
  Samelson/Whitehead products;
  anyons;
  fractional quantum Hall systems;
  flux quantization; 
  Hypothesis H%
}

\date{\today}

\dedicatory{
  \href{https://ncatlab.org/nlab/show/Center+for+Quantum+and+Topological+Systems}{\includegraphics[width=3.1cm]{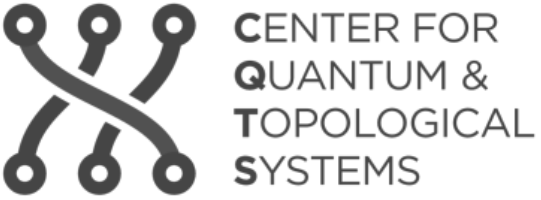}}
}

\begin{abstract}
We highlight that integer Heisenberg groups at level 2 underlie topological quantum phenomena: their group algebras coincide with the algebras of quantum observables of abelian anyons in fractional quantum Hall (FQH) systems on closed surfaces. Decades ago, these groups were shown to arise as the fundamental groups of the space of maps from the surface to the 2-sphere --- which has recently been understood as reflecting an effective FQH flux quantization in 2-Cohomotopy. Here we streamline and generalize this theorem using the homotopy theory of H-groups: We show that for $k \in \{1,2,4\}$, the non-torsion part of $\pi_1\, \Maps\bracket({(S^{2k-1})^2, S^{2k}})$ is an integer Heisenberg group of level 2, where we identify this level with 2 divided by the Hopf invariant of the generator of $\pi_{4k-1}(S^{2k})$. In particular, this result implies the existence of higher-dimensional analogs of FQH anyons in the cohomotopical completion of 11D supergravity (``Hypothesis H'').  
\end{abstract}

\maketitle

\tableofcontents

\section{Introduction}

Here we prove a curious result in elementary homotopy theory (on the latter cf. \cite{Whitehead1978}) with a striking relation to contemporary questions in quantum materials research, specifically in fractional quantum Hall systems (on the latter cf. \parencites{Stormer99}{Tong2016}{nLab:QuantumHallEffect}) relevant for questions in topological quantum computing (for which cf. \parencites{FreedmanKitaevLarsenWang2003}{Nayak2008}{Stanescu2020}{SatiValera2025}).

\subsection{Background}

\subsubsection{The original theorem}

Back in 1974, Hansen \cite{Hansen1974} investigated the fundamental groups of the space of maps $\Maps\,(-,-)$ (see \cref{MappingSpace}) from the torus $T^2 = (S^1)^2$ to the 2-sphere $S^2$. He found them --- in the connected component $\MapsComponent{n}\bracket({-,-})$ of winding number $n \in \mathbb{Z}$ ---  to be central extensions of $\mathbb{Z}^2$ by the cyclic group $\CyclicGroup{2 n}$:
\begin{equation}
  \label{HansenExtension}
  \begin{tikzcd}[sep=small]
  1
  \ar[r]
  &
  \CyclicGroup{2 n}
  \ar[rr]
  &&
  \pi_1
  \,
  \MapsComponent{n}\bracket({
    T^2
    ,\,
    S^2
  })
  \ar[rr]
  &&
  \mathbb{Z}^2
  \ar[r]
  \ar[
    ll, 
    shift right=5pt
  ]
  &
  1
  \mathrlap{\,.}
  \end{tikzcd}
\end{equation}
Since such central extensions are classified by the cohomology group
\begin{equation}
  \label{GroupOfCentralExtensionsOfZ2}
  H^2_{\mathrm{grp}}(\mathbb{Z}^2; \CyclicGroup{2n}) 
    \simeq
  H^2(T^2; \CyclicGroup{2n}) 
    \simeq 
  \CyclicGroup{2n}
  \mathrlap{\,,}
\end{equation}
Hansen's result determined these fundamental groups up to a \emph{level} $\ell_n \in \CyclicGroup{2 n}$. The groups arising this way (often considered only for unit level $\ell  = 1$) are also known as \emph{integer Heisenberg groups} (cf. \cite{nLab:IntegerHeisenbergGroup} and \cref{OnGroupTheory} below). 

In 1980, this problem was picked up by Larmore \& Thomas, who could show \cite[Thm. 1]{LarmoreThomas1980} that the level in \cref{HansenExtension} is in fact equal to $\ell = 2$, in all components.
In 2001, another proof of this fact was given by one of the authors  \cite[Prop. 1.5, Cor. 6.14]{Kallel2001}.
\footnote{
  These authors considered more generally the mapping space out of any closed oriented surface of genus $g$ (\cref{HigherGenusSurface}). But this turns out to be a fairly straightforward generalization of the toroidal case, $g = 1$, which we recover as \cref{Pi1OfMapsOutOfSurfaceIntoSphere} in \cref{OnGeneralizations} (where the statement is generalized further to higher dimensions). Focus on the case of the torus ($g = 1$) is also motivated by the physics application (further discussed in  \cref{OnOrdinaryFHQAnyonsIn2Cohomotopy}) where the experimental realization of FQH anyons becomes unfeasible for higher $g$, and where FQAH systems necessarily have  $g = 1$ (the \emph{Brillouin torus} of crystal momenta). 
}

\subsubsection{A hint of quantum physics}

We may make the following observations about this result:

While the abstract group theory literature typically considers the integer Heisenberg groups at level $\ell = 1$, the integer Heisenberg groups are subgroups of \emph{actual} Heisenberg groups $\mathrm{Heis}_3(\mathbb{R}, \periodicity)$ --- the hallmark structures of quantum mechanics (cf. \parencites[p. 7]{Rosenberg2004}) with fiber periodicity $\periodicity \in \mathbb{R}$ (cf. \cref{OrdinaryHeisenbergGroup} below) to which Heisenberg's name is referring here --- exactly at this level $\ell = 2$.
Here the cocycle that classifies the central extension is the restriction to the integers of the canonical \emph{symplectic form} $\omega$ on $\mathbb{R}^2$:
\begin{equation}
  \label{TheSymplecticForm}
  \begin{tikzcd}[
    row sep=-3pt
  ]
    \mathbb{R}^2
    \times
    \mathbb{R}^2
    \ar[
      r,
      "{ \omega }"
    ]
    &
    \mathbb{R}
    \\
    \bracket({
      (q,p),
      (q',p')
    })
    \ar[
      r,
      |->,
      shorten=5pt
    ]
    &
    q p' - p q'
    \mathrlap{.}
  \end{tikzcd}
\end{equation}
It is the two summands on the right of \cref{TheSymplecticForm} that, when restricting $\omega$ to a $\CyclicGroup{2n}$-valued group 2-cocycle on $\mathbb{Z}^2$, cause its class to be twice that of the generating class in \cref{GroupOfCentralExtensionsOfZ2}.

This quantum-mechanically natural level $\ell = 2$ for integer Heisenberg groups is the case we consider here by default, whence the result of \parencites{LarmoreThomas1980,Kallel2001} reads in our notation (cf. \cref{IntegerHeisenbergGroupAtLevel2,FinalTheorem}  below):
\begin{equation}
  \label{TheOriginalTheorem}
  \pi_1
  \,
  \MapsComponent{n}(
    T^2
    ,\,
    S^2
  )
  \,\simeq\,
  \mathrm{Heis}_3\bracket({\mathbb{Z},2n})
  \subset
  \mathrm{Heis}_3\bracket({\mathbb{R},2n})  
  \mathrlap{.}
\end{equation}

In \cref{OnTheTheorem} we give a new proof of this result \cref{TheOriginalTheorem}, shorter and more transparent than the previous arguments (by \cref{FundamentalGroupOfMapsFromTorusToSphere} below, invoking the classical theory of Samelson brackets, cf.  \cref{OnHGroupTheory}), and generalizing the statement to higher dimensions.

\subsubsection{Relation to FQH Anyons}
\label{OnRelationToFQHAnyons}

It is worth further expanding on this striking appearance of quantum mechanical structures in algebraic topology and homotopy theory (cf. \parencites{SS25-Srni}{SS25-FQH}): 

We recall that the group $C^\ast$-algebra of the ordinary Heisenberg group $\mathrm{Heis}_3(\mathbb{R},h)$ is essentially the \emph{Weyl algebra of quantum observables} on a 1-dimensional quantum system (cf. \cite[(3)]{Derezinski2006}) --- such as the famous but idealized ``particle on the line'' (or more realistically: a Josephson junction between superconductors in the \emph{Transmon regime}, where $q$ is the practically continuous phase difference of superconducting order parameters across the junction, cf. \cite{BenattietalEtAl2024}).

Similarly, the group algebra of the level $\ell = 2$ integer Heisenberg group $\mathrm{Heis}_3(\mathbb{Z}, \periodicity)$ is essentially the algebra of observables on a quantum system whose \emph{canonical coordinates} $q$ and \emph{canonical momenta} $p$ are constrained to discrete values. 

This is notably the case for quantum observables on \emph{anyons} in \emph{fractional quantum Hall systems} (FQH) on a torus (cf. \parencites[(4.9)]{WenNiu1990}[(4.14)]{IengoLechner1992}[(5.28)]{Tong2016}).
These FQH systems are 2-dimensional electron gases in a transverse magnetic field that is sufficiently strong and yet so effectively fine-tuned that there is an exact integer (generally: rational) multiple of magnetic flux quanta per electron (the inverse \emph{filling fraction}). On such a backdrop, the FQH anyons are elementary vortices in the electron gas associated with surplus magnetic flux quanta on top of the exact filling fraction. Here it is the \emph{flux quantization} of magnetic flux (now in the sense of \emph{discretization}, cf. \cite{SS25-Flux}) that makes the relevant Heisenberg group discrete.

What the discrete Heisenberg group here expresses is an intrinsic quantum mechanical \emph{braiding phase} $\zeta$ picked up by the quantum state $\vert \psi \rangle$ of the system as the anyons are moved around each other (cf. \cref{HeisenbergGroupCommutatorIllustrated}). 

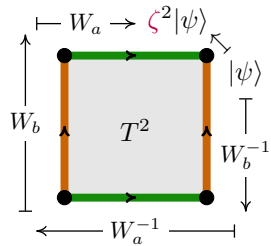
\begin{SCfigure}[1.95][htb]
\caption{
  \label{HeisenbergGroupCommutatorIllustrated}
  Illustrating the non-trivial group commutator 
  $
    [W_a, W_b]
    =
    \zeta^2
  $
  \cref{GroupCommutatorOfIntegerHeisAtL=2} in 
  the integer Heisenberg group 
  $\mathrm{Heis}_3(\mathbb{Z},h)$ at level $\ell = 2$ (\cref{IntegerHeisenbergGroupAtLevel2}).
  With the group elements understood as observables on FQH quantum systems (cf. \cref{OnOrdinaryFHQAnyonsIn2Cohomotopy}), the central element $\zeta$ is represented on their quantum states $\vert \psi \rangle$ by multiplication with a complex root of unity known as an \emph{anyon braiding phase} (cf. \parencites[(5.28)]{Tong2016}[Fig. 4]{SS25-FQAH}). 
}
\begin{tikzpicture}[scale=0.945,
  >={
    Computer Modern Rightarrow[
      length=4pt, width=4pt
    ]
  }
]

\draw[
  fill=lightgray,
]
  (0,0) rectangle (2,2);

\node at (1,.8) {%
  \clap{\smash{$T^2$}}%
};

\draw[
  line width=3,
  color=darkorange
] 
  (0,0) --
  (0,2);
\draw[
  ->,
  line width=1
]
  (0,1) --
  (0,1.01);

\draw[
  line width=3,
  color=darkorange
] 
  (2,0) --
  (2,2);
\draw[
  ->,
  line width=1
]
  (2,1) --
  (2,1.01);

\draw[
  line width=3,
  color=darkgreen
] 
  (0,2) --
  (2,2);
\draw[
  ->,
  line width=1
]
  (1,2) --
  (1.01,2);

\draw[
  line width=3,
  color=darkgreen
] 
  (0,0) --
  (2,0);
\draw[
  ->,
  line width=1
]
  (1,0) --
  (1.01,0);

\draw[
  fill=black
] 
(0,0) circle (.1);
\draw[
  fill=black
] 
(2,0) circle (.1);
\draw[
  fill=black
] 
(0,2) circle (.1);
(2,0) circle (.1);
\draw[
  fill=black
] 
(2,2) circle (.1);

\node at (2-.4,2+.5) {
  $\mathcolor{purple}{\zeta}^2\vert \psi \rangle$
};

\node[
  rotate=+135  
] at (2+.2, 2+.2) {%
  \clap{$\mapsto$}%
};

\node at (2.55,2-.25) {
  $\vert \psi \rangle$
};

\draw[
  line width=.5,
  |->,
]
  (2.55, 1.4) --
  node[
    xshift=0pt,
    scale=.9
  ] {\colorbox{white}{$
    W_b
      ^{-1}
  $}}
  (2.55,-.2);

\draw[
  line width=.5,
  |->,
]
  (2.4, -.47) --
  node[
    yshift=0pt,
    scale=.9
  ]{\colorbox{white}{$
    W_a^{-1}
  $}}
  (-.4, -.47);

\draw[
  line width=.5,
  |->,
]
  (-.55, -.2) --
  node[
    xshift=0pt,
    scale=.9
  ] {\colorbox{white}{$
    W_b
  $}}
  (-.55,2.3);

\draw[
  line width=.5,
  |->,
]
  (-.4, 2.45) --
  node[
    scale=.9,
  ]{\colorbox{white}{$
    W_a
  $}}
  (1, 2.45);
  
\end{tikzpicture}
\end{SCfigure}

Remarkably, the FQH anyon braiding phase $\zeta$ has been experimentally observed (for the moment not in toroidal but in planar electron gases, though) in recent years by various groups (starting with \cite{Nakamura2020,Bartolomei2020}, for further developments see \parencites{Veillon2024}{Ghosh2025}). This is of considerable technological interest, since it is the first and currently only observed case of fundamental anyonic \emph{topological quantum} phenomena (cf. \parencites{Stanescu2020}{Simon2023}). These are plausibly necessary (cf. \cite{DasSarma2022Hype}) for the future construction of \emph{quantum computers} (cf. \cite{Buyya2025Quantum}) of useful scale, namely for the engineering of \emph{topological quantum gates} (cf. \parencites{FreedmanKitaevLarsenWang2003}[\S 3]{MySS2024}{SatiValera2025}) which would be \emph{intrinsically} protected against the decohering noise that jeopardizes all quantum circuits.

However, even from a theoretical standpoint, there have remained open questions about the nature of FQH anyons (cf. \cite[\S A.1]{SS25-FQH}). The remarkable relation \cref{TheOriginalTheorem} to homotopy theory may help shed some light on these:

\subsubsection{Flux Quantization in Cohomotopy}

The above situation \cref{OnRelationToFQHAnyons} points to a curious physical interpretation \cite{SS25-FQH} of the result \cref{TheOriginalTheorem}: 

If we understand $\inlinetikzcd{S^2 \ar[r, hook] \& \! B \mathrm{U}(1)}$ as a deformation (namely the 3-skeleton) of the usual classifying space $B \mathrm{U}(1) \simeq B^2 \mathbb{Z}$ for magnetic charge, then \cref{TheOriginalTheorem} says that anyons in FQH systems may be understood as surplus magnetic flux quanta whose interaction with the 2D electron gas effectively deforms their \emph{flux quantization law} \cite[\S 3]{SS25-Flux} from ordinary cohomology $H^n(-;\mathbb{Z})$ to \emph{Cohomotopy} $\pi^n(-)$ (cf. \parencites[\S VII]{STHu59}[Ex. 2.7]{FSS23-Char}) in degree $n =2$:
\[
  \begin{aligned}
    H^n(-;\mathbb{Z})
    & \simeq
    \pi_0 
    \,
    \Maps\bracket({-, B^n \mathbb{Z}})
    \\
    \pi^n(-)
    \;
    & \defneq
    \pi_0 
    \,
    \Maps\bracket({-, S^n})
    \mathrlap{\,.}
  \end{aligned}
\]
This is a deformation in that the inclusion map $\inlinetikzcd{S^n \ar[r, hook] \& B^n \mathbb{Z}}$ induces isomorphism on homotopy groups in degree $\leq n$, which means that the above flux quantizing cohomologies agree on compact spaces of low dimension and then begin to disagree above that. Details on the comparison can be found in \cite{Grady2021}.

This suggests that FQH anyons may generally be understood as surplus magnetic flux quantized in 2-Cohomotopy (``Hypothesis h'', \parencites{SS25-AbelianAnyons}{SS25-WilsonLoops}{SS25-FQH}) which provides a new algebro-topological theory for the effective behavior of FQH anyons. For example, it predicts \parencites[Fig. D \& \S 3.5]{SS25-FQH}{SS26-Islands} that flux-expelling islands in FQH materials may support \emph{non-abelian} defect anyons, the realization of which is the holy grail of research on topological quantum materials.

In fact, flux quantization in Cohomotopy was first recognized as a physical possibility in higher-dimensional quantum field theory, specifically in 11-dimensional supergravity, whose ``C-field'' flux admits quantization in 4-Cohomotopy (``Hypothesis H'', \parencites[\S 2.5]{Sati2018}{FSS20-H}{FSS21-Hopf}{SS23-Mf}). We expand on these matters at the end in \cref{OnApplications}.

\subsection{Overview and Results}

This motivates us to ask whether an analogue of \cref{TheOriginalTheorem} may hold in higher dimensions. After briefly establishing the relevant context in \cref{OnPreliminaries}, our main result in \cref{OnTheTheorem}  is to show (\cref{FinalTheorem})  that, indeed:
\begin{equation}
  \pi_1
  \,
  \MapsComponent{[n]}\bracket({
    (S^3)^2
    ,\,
    S^4
  })
  \simeq
  \mathrm{Heis}_3(\mathbb{Z}, 0)
  \times^{\addedinline{\prime}}
  \CyclicGroup{12} \,,
\end{equation}
as well as
\begin{equation}
  \pi_1
  \,
  \MapsComponent{[n]}\bracket({
    (S^7)^2
    ,\,
    S^8
  })
  \simeq
  \mathrm{Heis}_3(\mathbb{Z}, 0)
  \times^{\addedinline{\prime}}
  \CyclicGroup{120}
\end{equation}
(where in both cases the components are labeled by $[n] \in \CyclicGroup{2}$ and the prime denotes the \emph{central product}.)

Remarkably, this result therefore implies the rigorous mathematical possibility of higher dimensional FQH-like anyon braiding phases in higher dimensions. We explain in \cref{On2DimensionalAnyonsVia4Cohomotopy} how this is indeed realized in the (physically hypothetical but mathematically well-defined) completion of \emph{11D supergravity} by cohomotopical flux-quantization, where we find 2-dimensional (``2-brane'') anyons in 6D, 
\footnote{
  Previous suggestions for anyonic effects in the higher dimensional braiding of 2-branes include \cite{Hartnoll2006} (which argues also in the context of 11D supergravity) and \cite{Feng2025} (which presents a kind of lattice model for anyonic membranes in 4D).
}
descending from 5-dimensional (``5-brane'') anyons in 11D.

Our novel move in proving these higher-dimensional generalizations is the intermediate 
\cref{FundamentalGroupOfMapsFromTorusToSphere}, bringing the classical theory of Samelson products on H-groups to bear on the problem (cf. \cref{SamelsonProductProperty}), which also serves to make the proof of the original 2-dimensional result much more transparent.

The proof strategy immediately generalizes to other situations, discussed in \cref{OnGeneralizations}.

We close in \cref{OnApplications} by expanding on the potential implications of this result on contemporary questions in the physics and engineering of anyonic quantum materials.

\section{Preliminaries}
\label{OnPreliminaries}

Before we come to the main theorem in \cref{OnTheTheorem}, here we briefly set up the context.

\subsection{Group theory}
\label{OnGroupTheory}

The ordinary \emph{Heisenberg group} (cf. \cite{Folland2025}) owes its name to its role as a cornerstone of basic quantum mechanics (cf. \parencites[p. 7]{Rosenberg2004}), where it reflects the fact that the elementary \emph{phase space} $\mathbb{R}^2$ with its canonical symplectic form \cref{TheSymplecticForm} becomes ``noncommutative'':
\begin{definition}[Ordinary Heisenberg group]
\label[definition]
 {OrdinaryHeisenbergGroup}
  For \emph{periodicity} $\periodicity \in \mathbb{R}$,
  the underlying smooth manifold 
  of the \emph{(reduced) Heisenberg group} (cf. \cite[p. 2, 25]{Folland2025})  is\footnote{
    In \cref{UnderlyingManifoldOfRealHeis3} the quotient group on the right is the additive group of real numbers quotiented by the subgroup generated by the element $\periodicity \in \mathbb{R}$, which up to isomorphism is:
    \[
      \CyclicReals{\periodicity}
      \simeq
      \begin{cases}
        \mathbb{R} & \text{if }
        \periodicity = 0
        \\
        \mathrm{U}(1) & 
        \text{otherwise.}
      \end{cases}
    \]
  }
  \begin{equation}
    \label{UnderlyingManifoldOfRealHeis3}
    \mathrm{Heis}_3\bracket({
      \mathbb{R}, 
      \periodicity
    })
    :=_{{}_{\mathrm{SmthMfd}}}
    \mathbb{R}^2 
      \times 
    \CyclicReals{\periodicity}
    \mathrlap{\,,}
  \end{equation}
  with generating elements to be denoted
  $$
    \left.
    \begin{aligned}
     W_a^q 
     & 
     := 
     \bracket({
       (q,0)
       ,
       [0]
     })
     \\
     W_b^p 
     & 
     := 
     \bracket({
       (0,p)
       ,
       [0]
     })
     \\
     \zeta^z 
       & :=
     \bracket({
       (0,0)
       ,
       [z]
     })
    \end{aligned}
    \right\}
    \in 
    \mathbb{R}^2 
      \times 
    \CyclicReals{\periodicity}
    \mathrlap{\,,}
  $$
  on which the only nontrivial group commutators are
  $$
    [
      W_a^q
      \,,\,
      W_b^p
    ]
    =
    \zeta^{2 q p}
    \mathrlap{.}
  $$

  Similarly, there are the \emph{higher-dimensional Heisenberg groups}, for $g \in \mathbb{N}_{\geq 1}$,
  \begin{equation}
    \label{UnderlyingManifoldOfHigherDimRealHeis}
    \mathrm{Heis}_{2g+1}\bracket({
      \mathbb{R}
      ,
      \periodicity
    })
    :=_{{}_{\mathrm{SmthMfd}}}
    \mathbb{R}^{2g}
    \times
    \CyclicReals{\periodicity}
    \mathrlap{\,,}
  \end{equation}
  with generators $\bracket({W^q_{a_i}, W^p_{b_i}})_{i =1}^g$ and $\zeta^z$ on which the only nontrivial group commutators are 
  $$
    \bracket[{
      W^q_{a_i}
      ,
      W^p_{b_i}
    }]
    =
    \zeta^{2qp}
    \,,\;\;\;\;
    \text{for $1 \leq i \leq g$}
    \mathrlap{.}
  $$
\end{definition}

\begin{definition}[Integer Heisenberg group at level $\ell=2$]
\label[definition]{IntegerHeisenbergGroupAtLevel2}
For periodicity $\periodicity \in 2\mathbb{Z}$, the underlying set is
\begin{equation}
  \label{UnderlyingSetOfIntegerHeisenbergGrp}
  \mathrm{Heis}_3\bracket({
    \mathbb{Z},
    \periodicity
  })
  :=_{\mathrm{Set}}
  \mathbb{Z}^2 
    \times 
  \CyclicGroup{\periodicity}
  \mathrlap{\,,}
\end{equation}
with generating elements
\begin{equation}
  \label{GeneratorsOfHeisenbergGrp}
  \left.
  \begin{aligned}
    W_a 
    & := 
    \bracket({
      (1,0), 
      [0]
    })
    \\
    W_b 
      & := 
    \bracket({
      (0,1), 
      [0]
    })
    \\
    \zeta 
      & := 
    \bracket({
      (0,0), 
      [1]
    })
  \end{aligned}
  \right\}
  \in
  \mathbb{Z}^2
    \times 
  \CyclicGroup
    {\periodicity}
  \mathrlap{\,,}
\end{equation}
on which the only non-trivial group commutator is
\begin{equation}
  \label{GroupCommutatorOfIntegerHeisAtL=2}
  [W_a, W_b] 
    = 
  \zeta^2
  \mathrlap{.}
\end{equation}

Similarly, there are the \emph{higher-dimensional integer Heisenberg groups}, for $g \in \mathbb{N}_{\geq 1}$,
\begin{equation}
\label{UnderlyingSetOfHigherDimIntegerHeis}
\mathrm{Heis}_{2g+1}\bracket({
  \mathbb{Z}
  ,
  \periodicity
})
:=_{{}_{\mathrm{SmthMfd}}}
\mathbb{Z}^{2g}
\times
\CyclicGroup{\periodicity}
\mathrlap{\,,}
\end{equation}
with generators $\bracket({W_{a_i}, W_{b_i}})_{i =1}^g$ and $\zeta$ on which the only nontrivial group commutators are 
\[
\bracket[{
  W_{a_i}
  ,
  W_{b_i}
}]
=
\zeta^{2}
\,,\;\;\;\;
\text{for $1 \leq i \leq g$}
\mathrlap{.}
\]
\end{definition}
Hence we have evident subgroup inclusions
\[
  \inlinetikzcd{
    \mathrm{Heis}_3\bracket({
      \mathbb{Z},
      \periodicity
    })
    \ar[r, hook]
    \&
    \mathrm{Heis}_3\bracket({
      \mathbb{R},
      \periodicity
    })
  }
  \mathrlap{\,.}
\]
\begin{remark}[Other levels]
Over the integers, non-isomorphic groups are obtained by instead taking the above group commutator \cref{GroupCommutatorOfIntegerHeisAtL=2} to be
$$
  [W_a, W_b] 
    = 
  \zeta^{\ell},
$$
for some 
$$
  \ell
    \in 
  \mathbb{Z}
  \simeq
  H^2_{\mathrm{Grp}}(
    \mathbb{Z}^2
    ;\,
    \mathbb{Z}
  )
  \mathrlap{\,,}
$$ 
which we call the \emph{level}, since it is the central extension class. Authors in group theory are usually concerned with the case $\ell = 1$ (cf. \cite{nLab:IntegerHeisenbergGroup}); but for our purpose, the case $\ell = 2$ is of paramount importance, so we take this to be understood by default \cref{GroupCommutatorOfIntegerHeisAtL=2}.
\end{remark}

\begin{remark}[Planck constant and Chern-Simons level]
\label[remark]
{PlancConstantAndCSLevel}
  In physics, the Hilbert space $\HilbertSpace$ of quantum states forms a representation of the Heisenberg group, where the central generator $\zeta := [1] \in \CyclicGroup{\periodicity} \subset \CyclicReals{\periodicity}$ acts as a complex phase:
  \[
    \zeta 
    \mapsto 
    \exp\bracket({
      \mathrm{i}\pi h 
    })
  \]  
  with $h \in \mathbb{R}$.
  In representations of $\mathrm{Heis}_3\bracket({\mathbb{R},\periodicity})$ for general quantum systems, this $h$ is to be understood as \emph{Planck's constant} (cf. \cite[p. 6]{Folland2025}).  
  In representations of $\mathrm{Heis}_3\bracket({\mathbb{Z}, \periodicity})$ for FQH systems the same parameter plays the role of the \emph{filling fraction} $\nu$, and its inverse $K = 1/\nu$ (when integral) is the \emph{Chern-Simons level} (cf. \cite[p. 5]{SS25-FQH}). 

  Often this is considered for the unreduced Heisenberg groups with $\periodicity = 0$. Otherwise, for the reduced Heisenberg groups (cf. \cite[p. 25]{Folland2025}) the relation $\zeta^{\periodicity} = \mathrm{e} \in \CyclicGroup{\periodicity} \subset \CyclicReals{\periodicity}$ induces  the constraint
  \[
    2 \pi \mathrm{i} \, C/K
    \in 
    2\pi \mathrm{i}
    \mathbb{Z}
    \mathrlap{\,.}
  \]
  Therefore a nonvanishing periodicity parameter $\periodicity$, is a bound and restriction on the admissible CS level $K$ (inverse Planck constant). 
\end{remark}

\subsection{H-Group theory}
\label{OnHGroupTheory}

We need the following concepts and results from the theory of homotopy groups (for background see \cref{OnSomeHomotopyTheory}) which are classical but not always readily citable (such as the crucial \cref{SamelsonProductProperty} below).

\begin{definition}[{cf. \cite[\S X.5]{Whitehead1978}}]
\label[definition]{SamelsonProduct}
The \emph{Samelson product} $[ -,-]_{\mathrm{Sam}}$ on a loop space $\Omega X$ (\cref{LoopSpaceAsHGroup}, and generally on an H-group, \cref{HGroup}, and from there induced on its homotopy groups) is (up to homotopy) the \emph{H-group commutator} descended from the Cartesian product to the smash product:
\begin{equation}
  \label{SamelsonProductOnLoopSpace}
  \begin{tikzcd}[
    row sep=0pt,
    column sep=25pt
  ]
  (\ell_2, \ell_1)
  \ar[
      rr,
      shorten=5pt,
      |->
    ]
    &&
    (\ell_2 
      \star 
    \ell_1)
      \star
    (\overline{\ell_2} \star
    \overline{\ell_1})
    \\
    (\Omega X)
    \times 
    (\Omega X)
    \ar[
      rr,
      "{
        [-,-]
      }"
    ]
    \ar[
      d,
      ->>
    ]
    &&
    \Omega X
    \mathrlap{\,.}
    \\[15pt]
    (\Omega X)
    \wedge
    (\Omega X)
    \ar[
      urr,
      "{
        [ -,- ]_{\mathrm{Sam}}
      }"{swap, sloped, pos=.4}
    ]
  \end{tikzcd}
\end{equation}
\end{definition}

\begin{definition}[cf. {\cite[p. 176-7]{FHT2000}}]
\label[definition]{WhiteheadProduct}
  For $X \in \mathrm{TopSp}^\ast$, the \emph{Whitehead} product on its homotopy group (in degrees $n_1, n_2 \in \mathbb{N}_{\geq 1}$) is the graded commutative map
  \[
    \begin{tikzcd}
      \pi_{n_1}(X)
      \times
      \pi_{n_2}(X)
      \ar[
        rr,
        "{ [-,-]_{\mathrm{Wh}} }"
      ]
      &&
      \pi_{n_1 + n_2 - 1 }(X)
      \mathrlap{\,,}
    \end{tikzcd}
  \]
  given on a pair of representatives
  \[
    [\,
    \inlinetikzcd{
      \phi_i 
        :
      S^{n_i}
      \ar[r]
      \&
      X
    }
    \,]
    \in
    \pi_{n_i}(X)
    \,,
    \;\;
    i \in \{1,2\}
  \]
  by the homotopy class of the top composite map of the following diagram:
  \[
    \begin{tikzcd}[
      column sep=55pt
    ]
      S^{n_1 + n_2 - 1}
      \ar[
        r,
        "{ f_{n_1, n_2} }"
      ]
      \ar[
        rr,
        uphordown,
        "{
          [\phi_1,\phi_2]_{{}_{\mathrm{Wh}}}
        }"{description}
      ]
      \ar[d]
      \ar[
        dr,
        phantom,
        "{ \ulcorner }"{pos=.9}
      ]
      &
      S^{n_1} \vee S^{n_2}
      \ar[
        r,
        "{ (\phi_1, \phi_2) }"
      ]
      \ar[d]
      &
      X
      \\
      D^{n_1 + n_2}
      \ar[
        r
      ]
      &
      S^{n_1} \times S^{n_2}
      \mathrlap{\,,}
    \end{tikzcd}
  \]
  where $f_{n_1, n_2}$ is the attaching map for the cell attachment to the wedge sum of spheres that gives the product of spheres, as shown.

  For $n_1, n_2 \geq 2$ the Whitehead bracket is bilinear, for $n_1 = 1$ and $n_2 \geq 2$ it reduces to the $\pi_1$-action minusss the identity, and for $n_1 = 1 = n_2$ it reduces to the group commutator.
\end{definition}

\begin{proposition}[cf. {\cite[Thm. 7.10 on p. 476]{Whitehead1978}}]
\label{WhiteheadBracketRelatingToSamelson}
For $X \in \mathrm{TopSp}^\ast$, the  \emph{Whitehead bracket} $[-,-]_{\mathrm{Wh}}$ on $\pi_{\bullet+1}(X)$ 
\textup{(\cref{WhiteheadProduct})}
is given by the Samelson product on $\pi_\bullet(\Omega X)$ \textup{(\cref{SamelsonProduct})}
as:
\begin{equation}
  \label{WhiteheadProductAsSamelsonProduct}
  \widetilde{%
    [\alpha,\alpha_2]%
  }_{\mathrm{Wh}}
  =
  (-1)^{\mathrm{deg}(\alpha_1)} 
  \bracket[{
    \widetilde{\alpha}_1
    ,\,
    \widetilde{\alpha}_2 
    \,
  }]_{\mathrm{Sam}}
  \mathrlap{\,,}
\end{equation}
where 
$
  \begin{tikzcd}[]
    \pi_{\bullet + 1}(X)
    \ar[
      <->,
      r, 
      "{\sim}"{swap},
      "{ \widetilde{(-)} }"
    ]
    &
    \pi_\bullet(\Omega X)
  \end{tikzcd}
$ is induced from the hom-isomorphism \cref{SuspensionLoopAdjunction}:
$$
  \begin{tikzcd}[
    column sep=50pt
  ]
    \pi_0\,
    \PointedMaps\bracket({
      S^{n_1}
      ,\,
      \Omega X
    })
    \times
    \pi_0\,
    \PointedMaps\bracket({
      S^{n_2}
      ,\,
      \Omega X
    })
    \ar[
      r,
      "{
        [-,-]_{\mathrm{Sam}}
      }"
    ]
    \ar[
      d,
      <->,
      "{ 
        \widetilde{(-)} 
        \times
        \widetilde{(-)} 
      }"
    ]
    &
    \pi_0\,
    \PointedMaps\bracket({
      S^{n_1 + n_2}
      ,\,
      \Omega X
    })
    \ar[
      d,
      <->,
      "{ \widetilde{(-)} }"
    ]
    \\
    \pi_0\,
    \PointedMaps(
      S^{n_1 + 1}
      ,\,
      X
    )
    \times
    \pi_0\,
    \PointedMaps(
      S^{n_2 + 1}
      ,\,
      X
    )
    \ar[
      r,
      "{
        \pm [-,-]_{\mathrm{Wh}}
      }"
    ]
    &
    \pi_0\,
    \PointedMaps(
      S^{n_1 + n_2 + 1}
      ,\,
      X
    )
    \mathrlap{.}
  \end{tikzcd}
$$
\end{proposition}
\begin{lemma}[cf. {\cite[Thm. 8.20 on p. 485]{Whitehead1978}}\footnote{Beware that \cite{Whitehead1978} writes ``$E$'' for the suspension functor, an old notational convention introduced there on p. 369.}]
\label[lemma]
  {SuspensionOfWhiteheadVanishes}
  The suspension \cref{Suspension} of any Whitehead bracket \eqref{WhiteheadProductAsSamelsonProduct} vanishes:
  \begin{equation}
    \begin{tikzcd}[row sep=-2pt]
      \pi_n(S^m)
      \ar[r, "{ \Sigma }"]
      &
      \pi_{n+1}(S^{m+1})
      \\
      {[\alpha, \beta]}_{\mathrm{Wh}}
      \ar[
        r,
        |->, 
        shorten=5pt
      ]
      &
      0
      \mathrlap{\,.}
    \end{tikzcd}
  \end{equation}
\end{lemma}

\begin{proposition}
\label[proposition]
 {SamelsonProductProperty}
For $k \in \{1,2,4\}$, the following diagram commutes up to homotopy:
\begin{equation}
  \label{SamelsonRelationAsCommutingDiagram}
  \begin{tikzcd}[row sep=8pt, column sep=55pt]
    &
    (\Omega S^{2k}) 
      \wedge 
    (\Omega S^{2k})
    \ar[
      dr,
      "{
        [-,-]_{\mathrm{Sam}}
      }"{sloped}
    ]
    \\
    S^{4k-2}
    \ar[
      ur,
      "{
        \widetilde{\mathrm{id}_{2k}}
        \,\wedge\,
        \widetilde{\mathrm{id}_{2k}}
      }"{sloped, pos=.4}
    ]
    \ar[
      rr,
      "{
        \widetilde{
        2h_{\mathbb{K}}
        \addedinline{
        -
        \kappa_{\mathbb{K}}
        }
        }
      }"
    ]
    &&
    \Omega S^{2k}
    \mathrlap{\,,}
  \end{tikzcd}
\end{equation}
where 
\begin{enumerate}
\item
$ 
  h_{_{\mathbb{K}}}
$ denotes the $\mathbb{K}$-Hopf fibration, the generator of the $\mathbb{Z}$-subgroup

\item
\addedinline{
$\kappa_{_{\mathbb{K}}}$ denotes a generator of the torsion subgroup ($0$ if $k = 1$)}
\end{enumerate}
of 
$\pi_{4k-1}\bracket({S^{2k}})$ \cref{SomeUnstableHomotopyGroupsOfSpheres},
for
\[
  \mathbb{K}
  \defneq
  \begin{cases}
    \mathbb{C} & \textup{ if } k = 1
    \\
    \mathbb{H} & \textup{ if } k = 2
    \\
    \mathbb{O} & \textup{ if } k = 4
    \mathrlap{\,.}
  \end{cases}
\]

\end{proposition}
\begin{proof}
The Whitehead bracket of the generator $[\mathrm{id}_{2k}] := 1 \in \mathbb{Z} \simeq \pi_{2k}(S^{2k})$ with itself has Hopf invariant 2 (cf. \cite[Thm. 2.5 on p. 495]{Whitehead1978}):
\[
  H\bracket({
    \bracket[{\mathrm{id}_{2k}, \mathrm{id}_{2k}}]
      _{\mathrm{Wh}}
  })
  =
  2
  \,.
\]
Now the integer valued Hopf invariant $H(-)$ is a homomorphism \cite[(2.1) on p. 495]{Whitehead1978}, whence torsion elements have vanishing Hopf invariant, while, famously, $H(h_{\mathbb{K}}) = 1$ \cite{Hopf1935}. 
\addedinline{Therefore it follows that the element in question is of the form
\begin{equation}
\label
{DecomposingTheWhiteheadBracket}
  \bracket[{
    \mathrm{id}_{2k},
    \mathrm{id}_{2k}
  }]_{\mathrm{Wh}}
  =
  2 [h_{_{\mathbb{K}}}] - t_{\mathbb{K}}
\end{equation}
for some torsion element $t_{\mathbb{K}}$.
}

\addedinline{
The value of $t_{\mathbb{K}}$ does not appear to be citable from classical literature (while the formula superficially appears in \parencites[(1.3)]{Oguchi1964}\cite[(2.1)]{MiyauchiMukai2017}, there it only refers to the 2-primary projection, following \cite{Toda1962}), but may be determined as follows:}
\begin{description}

\item[$k = 1$]
\addedinline{
Here the torsion subgroup is trivial and hence $t_{\mathbb{C}} = 0$. 
}

\item[$k = 2$] 
\addedinline{
Since the suspension of any Whitehead product vanishes (\cref{SuspensionOfWhiteheadVanishes}), the suspension of \cref{DecomposingTheWhiteheadBracket} yields
$\Sigma t_{\mathbb{H}} = 2 \Sigma [h_{\mathbb{H}}]$. Since $\Sigma [h_{\mathbb{H}}]$ is a generator of $\pi_8\bracket({S^5}) \simeq \CyclicGroup{24}$ (by \cite[Thm. 1.5 \& \S 10]{Adams1966}, cf. \cite[Thm. 1.1.13]{Ravenel2004}), this implies that $t_{\mathbb{H}}$ has order a multiple of $12$, hence must be a generator of the ambient torsion subgroup whose order is just $12$ \cref{SomeUnstableHomotopyGroupsOfSpheres}.
}

\item[$k = 4$] 
\addedinline{
This is directly analogous to the previous case, now using that $\Sigma [h_{\mathbb{O}}]$ generates $\pi_{16}(S^9) \simeq \CyclicGroup{240}$ while the torsion subgroup of $\pi_{15}(S^8)$ is $\CyclicGroup{120}$ \cref{SomeUnstableHomotopyGroupsOfSpheres}.
}

\end{description}

In summary:
\begin{equation}
\label{WhiteheadSquareRootsOfHopfFibrations}
\begin{aligned}
  \bracket[{
    [\mathrm{id}_2],
    [\mathrm{id}_2]
  }]_{\mathrm{Wh}}
  & 
    = 2 [h_{\mathbb{C}}]
  \\
  \bracket[{
    [\mathrm{id}_4],
    [\mathrm{id}_4]
  }]_{\mathrm{Wh}}
  & 
    = 2 [h_{\mathbb{H}}] 
    \addedinline{
      - \kappa_{_{\mathbb{H}}} 
    }
  \\
  \bracket[{
    [\mathrm{id}_8],
    [\mathrm{id}_8]
  }]_{\mathrm{Wh}}
  & 
    = 2 [h_{\mathbb{O}}]
    \addedinline{
      - \kappa_{_{\mathbb{O}}}
    }
  \mathrlap{\,,}
\end{aligned}
\end{equation}
which implies the claim.
\end{proof}

\begin{remark}
  We will see that the factor of 2 appearing in \cref{WhiteheadSquareRootsOfHopfFibrations} is the origin both of the level of the Heisenberg group being 2 (this is the content of our new proof \cref{FundamentalGroupOfMapsFromTorusToSphere}), as well as its ``Planck constant'' being $2n$ when $k = 1$ (by \cref{FundamentalGroupsOfMapsS2ToS2} below, which was observed before). 
\end{remark}

\section{The Theorem}
\label{OnTheTheorem}

We compute here the fundamental groups of the mapping spaces from $(S^{2k-1})^2$ to $S^{2k}$, for $k \in \{1,2,4\}$, using the H-group theory of \cref{OnHGroupTheory} and find them to essentially be the integer Heisenberg groups of \cref{OnGroupTheory}.

\subsection{The level two}

This section contains our main new observation (proof of \cref{FundamentalGroupOfMapsFromTorusToSphere} below).
First, we make these basic observations:
\begin{remark}[Cell structure of the squared sphere]
\begin{enumerate}
 
 \item The canonical CW-complex structure of $(S^{2k-1})^2$ has, besides a 0-cell, two $(2k-1)$-cells and a single $(4k-2)$-cell, attached via the Whitehead product of their identity maps (by \cref{WhiteheadProduct}), as shown by the following pushout square on the left:
  \begin{equation}
    \label{CWComplexStructureOnSquaredSphere}
    \begin{tikzcd}[column sep=45pt]
      S^{4k-3} 
      \ar[
        rr,
        "{ 
          [
            \mathrm{id}_{2k-1}, 
            \mathrm{id}_{2k-1}
          ]_{\mathrm{Wh}} 
        }"
      ]
      \ar[
        d,
        hook
      ]
      &{}
      \ar[
        dr,
        phantom,
        "{ \ulcorner }"{pos=.9}
      ]
      &
      S^{2k-1} \vee S^{2k-1}
      \ar[
        d,
        hook
      ]
      \ar[r]
      \ar[
        dr,
        phantom,
        "{ \ulcorner }"{pos=.9}
      ]
      &
      \ast
      \ar[d]
      \\
      D^{4k-2}
      \ar[
        rr,
        "{
          i_{4k-2}
        }"
      ]
      &&
      S^{2k-1} \times S^{2k-1}
      \ar[
        r,
        "{
          \mathrm{pr}_c
        }"
      ]
      \ar[
        d,
        shift left=5pt,
        "{
          \mathrm{pr}_{b}
        }"
      ]
      \ar[
        d,
        shift right=5pt,
        "{
          \mathrm{pr}_{a}
        }"{swap}
      ]
      &
      S^{4k-2}
      \\[+5pt]
      &&
      S^{2k-1}
      \mathrlap{\,.}
    \end{tikzcd}
  \end{equation}
  The pushout square on the right forms the smash product of the two $(2k-1)$-cells and thereby exhibits a canonical map denoted $\mathrm{pr}_c$. The two projections onto the $(2k-1)$-sphere factors are shown at the bottom of the above diagram, which we denote by $\mathrm{pr}_a$ and $\mathrm{pr}_b$, respectively.

\item  
  Note that the following diagram commutes:
  \begin{equation}
    \label{IdentityBetweenProjections}
    \begin{tikzcd}[
      row sep=5pt, 
      column sep=40pt
    ]
      & 
      (S^{2k-1})^2\wedge (S^{2k-1})^2
      \ar[
        dr,
        "{ 
          \mathrm{pr}_a
          \wedge
          \mathrm{pr}_b
        }"{sloped}
      ]
      \\
      (S^{2k-1})^2
      \ar[
        rr,
        "{ \mathrm{pr}_c }"
      ]
      \ar[
        ur,
        "{ \Delta }"{sloped, pos=.4}
      ]
      &&
      S^{4k-2}
      \mathrlap{\,.}
    \end{tikzcd}
  \end{equation}
  (This follows by inspection; alternatively, it is a special case of \cref{SmashOfProjectionMaps} below).

\item  
  The suspension of these projection maps yields the three coprojections out of the \emph{stable splitting} (using \cref{SuspensionOfWhiteheadVanishes} or \cite[Prop. 4I.1]{Hatcher2002}) of the squared sphere:
\begin{equation}
  \label{StableSplittingOfSuspendedSquaredSphere}
  \begin{tikzcd}[
    column sep=-2pt
  ]
  \Sigma(S^{2k-1})^2
  &
  \underset{
  \mathrm{hmtp}
  }{\simeq}
  &
  S^{2k}_a 
  \ar[
   d,
   "{ \Sigma \mathrm{pr}_a }"
   {description}
  ]
    &\vee& 
  S^{2k}_b 
  \ar[
   d,
   "{ \Sigma \mathrm{pr}_b }"{description}
  ]
    &\vee& 
  S^{4k-1}
  \ar[
   d,
   "{ \Sigma \mathrm{pr}_c }"{description}
  ]
  \\
  &&
  S^{2k}_a
  &&
  S^{2k}_b
  &&
  S^{4k-1}
  \mathrlap{\,.}
  \end{tikzcd}
\end{equation}
Notice that this homotopy equivalence does \emph{not} respect the H-cogroup structure on suspensions (from \cref{LoopSpaceAsHGroup}).
\end{enumerate} 
\end{remark}

\begin{proposition}
\label[proposition]
  {OnConnectedComponentsOfTheMappingSpace}
  The connected components of the pointed mapping space are:
  \begin{equation}
    \label{ConnectedComponentsOfMappingSpace}
    \pi_0\,
    \PointedMaps\bracket({
      (S^{2k-1})^2
      ,\,
      S^{2k}
    })
    \simeq
    \begin{cases}
      \mathbb{Z}
      & \text{if } k = 1
      \\
      \CyclicGroup{2}
      & \text{if } k = 2
      \\
      \CyclicGroup{2}
      & \text{if } k = 4 
      \mathrlap{\,.}
    \end{cases}
  \end{equation}
\end{proposition}
\begin{proof}
  Generally, maps to a $(2k-1)$-connected space like $S^{2k}$ factor through the quotient of the domain by its $(2k-1)$-skeleton. 
  But, by the CW-structure \cref{CWComplexStructureOnSquaredSphere},
  we have:
  $$
    (S^{2k-1})^2
    \big/
    \mathrm{sk}_{2k-1}
    (S^{2k-1})^2
    \simeq
    (S^{2k-1})^2
    \big/
    (S^{2k-1} \vee S^{2k-1})
    \simeq
    S^{4k-2}
    \mathrlap{\,,}
  $$
  %
  %
  whence:
  $$
    \pi_0\,
    \PointedMaps\bracket({
      (S^{2k-1})^2
      ,\,
      S^{2k}
    })
    \simeq
    \pi_0\,
    \PointedMaps\bracket({
      S^{4k-2}
      ,\,
      S^{2k}
    })    
    \defneq
    \pi_{4k-2}(S^{2k})
    \mathrlap{\,.}
  $$
  With the standard homotopy group of spheres \cref{SomeUnstableHomotopyGroupsOfSpheres}, this proves the claim.
\end{proof}

Our central observation now is the proof of the following:
\begin{theorem}
  \label[theorem]
    {FundamentalGroupOfMapsFromTorusToSphere}
  For $k \in \{1,2,4\}$,
  we have a group isomorphism
  \begin{equation}
  \label{IdentifyingFundamentalGroupOfMapsFromTorusToSphere}
    \pi_1 
    \,
    \PointedMapsComponent{0}\bracket({
      (S^{2k-1})^2
      ,\,
      S^{2k}
    })
    \simeq
    \mathrm{Heis}_3(\mathbb{Z}, 0)
    \times^{
      \addedinline{\prime}
    }
    \begin{cases}
      1 
      & \text{\rm if } k = 1
      \\
      \CyclicGroup{12} 
      & \text{\rm if } k = 2
      \\
      \CyclicGroup{120} 
      & \text{\rm if } k = 4
      \mathrlap{\,,}
    \end{cases}
  \end{equation}
  \addedinline{
  where the primed product on the right denotes the ordinary product for $k = 1$ and otherwise the central product
  \footnote{
    \addedinline{
    In terms of the definition of \emph{central product groups} as in \cite[p. 157]{DummitFoote2003}, our primed product is 
    \[
      \mathrm{Heis}_3(\mathbb{Z}, 0) \times^\prime A 
        \simeq 
      \bracket({
        \mathrm{Heis}_3(\mathbb{Z}, 0)
        \times 
        \bracket({
          \mathbb{Z} \times A
        })
      })
      \big/
      \bracket\langle 
      {
        \bracket({
          \bracket({
            0,0
          })
          -1
        }), 
        \bracket({
          1, [-1]
        })
      }
      \rangle
      \,.
     \]
     }
  }
  with 
  \begin{equation}
    \label{TheCentralProduct}
    \bracket[{W_a,W_b}]
    =
    \zeta^{2} \cdot [-1]
    \mathrlap{\,.}
  \end{equation}
  }
\end{theorem}
\begin{proof}
  First, we observe the bijection \cref{IdentifyingFundamentalGroupOfMapsFromTorusToSphere} on underlying sets:
  \begin{align*}
    \pi_1\, \PointedMapsComponent{0}\bracket({
      (S^{2k-1})^2
      ,\,
      S^{2k}
    })
    &
    \simeq
    \pi_0\, \PointedMaps\bracket({
      \Sigma(S^{2k-1})^2
      ,\,
      S^{2k}
    })    
    &&
    \proofstep{%
      by \cref{IncaranationsOfFundGrpdOfMappingSpace}
    }
    \\
    &
    \simeq
    \pi_0\, \PointedMaps\bracket({
      S^{2k}
        \vee
      S^{2k}
        \vee
      S^{4k-1}
      ,\,
      S^{2k}
    })
    &&
    \proofstep{%
      by \cref{StableSplittingOfSuspendedSquaredSphere}
    }
    \\
    &
    \simeq
    \pi_{2k}(S^{2k})^2
    \times
    \pi_{4k-1}(S^{2k})
    &&
    \proofstep{%
      by \cref{PointedMapsOutOfWedgeSum,HomotopyGroupsViaMappingSpace}
    }
    \\
    &
    \simeq
    \mathbb{Z}^2
    \times
    \mathbb{Z} \times 
    \begin{cases}
      1 & \text{if } k = 1
      \\
      \CyclicGroup{12} & \text{if } k = 2
      \\
      \CyclicGroup{120} & \text{if } k = 4
    \end{cases}
    &&
    \proofstep{
      by \cref{SomeUnstableHomotopyGroupsOfSpheres}
    }.
  \end{align*}
  (This is a sequence of bijections of sets, but not of homomorphisms of groups, since the stable splitting used in the second step is not an H-cogroup homomorphism.) \\
Therefore, the candidate Heisenberg group generators \cref{GeneratorsOfHeisenbergGrp} are represented by the projectors \cref{StableSplittingOfSuspendedSquaredSphere} and the homotopy group generators \cref{SomeUnstableHomotopyGroupsOfSpheres} as:
\begin{equation}
  \begin{tikzcd}[
    column sep=-2pt,
    row sep=0pt
  ]
    W_{a/b}
    &:=&
    \big[
      \Sigma(S^{2k-1})^2
      \ar[
        r,
        "{ 
          \Sigma\,\mathrm{pr}_{a/b} 
        \;}"
      ]
      &[30pt]
      S^{2k}
      \ar[r, "{ \mathrm{id} }"]
      &[25pt]
      S^{2k}
    \big]
    \\
    \zeta
    &:=&
    \big[
      \Sigma(S^{2k-1})^2
      \ar[
        r,
        "{ 
          \Sigma\,\mathrm{pr}_{c} 
        }"
      ]
      &[30pt]
      S^{4k-1}
      \ar[
        r, 
        "{ h_{\mathbb{K}} }"
      ]
      &[25pt]
      S^{2k}
    \big]
    \mathrlap{\,,}
  \end{tikzcd}
\end{equation}
and adjointly as:
\begin{equation}
  \begin{tikzcd}[
    column sep=1pt,
    row sep=0pt
  ]
    \widetilde{W}_{a/b}
    &=&
    \big[
      (S^{2k-1})^2
      \ar[
        r, 
        "{ \mathrm{pr}_{a/b} \,}"
      ]
      &[25pt]
      S^{2k-1}
      \ar[
        r,
        "{ \widetilde{\mathrm{id}} }"
      ]
      &[20pt]
      \Omega S^{2k}
    \big]
    \\
    \widetilde{\zeta}
    &=&
    \big[
    (S^{2k-1})^2
    \ar[
      r,
      "{ \mathrm{pr}_c }"
    ]
    &
    S^{4k-2}
    \ar[
      r,
      "{
        \widetilde{h_{\mathbb{K}}}
      \;}"
    ]
    &
    \Omega S^{2k}
    \big]
    \mathrlap{\,,}
  \end{tikzcd}
\end{equation}
while the generator of the further torsion group factor is represented as:
$$
  \begin{tikzcd}[
    column sep=1pt
  ]
    R
    :=
    \big[
    \Sigma (S^{2k-1})^2
    \ar[
      r,
      "{ \Sigma\,\mathrm{pr}_c \;}"
    ]
    &[25pt]
    S^{4k-1}
    \ar[
      r,
      "{ r_{\mathbb{K}} }"
    ]
    &[20pt]
    S^{2k}
    \big]
    \mathrlap{\,.}
  \end{tikzcd}
$$
  Then, by \cref{GroupStructureOnComponents,SamelsonProduct}, the group commutator $[W_a, W_b]$ is represented by the outer part of the following diagram:
  \begin{equation}
    \label{FirstDiagramForDerivingHeisenberg}
    \begin{tikzcd}[
      column sep=45pt
    ]
      (S^{2k-1})^2
      \ar[
        d,
        "{ \Delta }"{swap, pos=.4}
      ]
      \ar[
        rr,
        "{
          \widetilde{%
            [W_a, W_b]%
          }
        }"
      ]
      \ar[
        dr,
        "{
          \mathrm{pr}_c
        }"{description}
      ]
      &&
      \Omega S^{2k}
      \\
      (S^{2k-1})^2
      \wedge
      (S^{2k-1})^2
      \ar[
        rr,
        downhorup,
        "{
          \widetilde{W}_a
          \,\wedge\,
          \widetilde{W}_b
        }"{description}
      ]
      \ar[
        r,
        "{
          \mathrm{pr}_a
          \,\wedge\, 
          \mathrm{pr}_b
        }"{swap}
      ]
      &
      S^{2k-1} 
        \wedge 
      S^{2k-1}
      \ar[
        r,
        "{
          \widetilde{\mathrm{id}_{2k}}
          \,\wedge\,
          \widetilde{\mathrm{id}_{2k}}
        \;}"{swap}
      ]
      \ar[
        ur,
        "{
          \widetilde{
            2h_{\mathbb{K}}
            \addedinline{
              -
              \kappa_{_{\mathbb{K}}}
            }
          }
        }"{description}
      ]
      &
      (\Omega S^{2k})
      \wedge
      (\Omega S^{2k})
      \mathrlap{\,.}
      \ar[
        u,
        "{
          [-,-]_{\mathrm{Sam}}
        }"{swap}
      ]
    \end{tikzcd}
  \end{equation}
  Inside this diagram, we identify (as shown)
  \begin{enumerate}
  \item
  the left diagonal map by \cref{IdentityBetweenProjections},
  \item
  the right diagonal map by \cref{SamelsonRelationAsCommutingDiagram} from \cref{SamelsonProductProperty}.
  \end{enumerate}
  Hence, the top inner triangle in the diagram commutes and proves 
  $
    [W_a, W_b] 
      = 
    \zeta^2 
    \addedinline{\cdot[-1]}
  $.

  Analogously, $[W_{a/b}, \zeta]$ is represented by the outer part of the following diagram
  \begin{equation}
    \label{SecondDiagramForDerivingHeisenberg}
    \begin{tikzcd}[
      column sep=45pt
    ]
      (S^{2k-1})^2
      \ar[
        d,
        "{ \Delta }"{swap}
      ]
      \ar[
        dr,
        "{ 0 }"{description}
      ]
      \ar[
        rr,
        "{
          \widetilde{[W_{a/b}, \zeta]}
        }"
      ]
      &&
      \Omega S^{2k}
      \\
      (S^{2k-1})^2
      \wedge
      (S^{2k-1})^2
      \ar[
        rr,
        downhorup,
        "{
          \widetilde{W}_{a/b}
          \,\wedge\,
          \widetilde{\zeta}
        }"{description}
      ]
      \ar[
        r,
        "{
          \mathrm{pr}_{a/b}
          \,\wedge\, 
          \mathrm{pr}_c
        }"{swap}
      ]
      &
      S^{2k-1}
        \!\wedge\!
      S^{4k-2}
      \ar[
        r,
        "{
          \widetilde{\mathrm{id}_{2k}}
          \,\wedge\,
          \widetilde{h_{\mathbb{K}}}
        }"{swap}
      ]
      &
      (\Omega S^{2k})
      \wedge
      (\Omega S^{2k})
      \mathrlap{\,.}
      \ar[
        u,
        "{
          [-,-]_{\mathrm{Sam}}
        }"{swap}
      ]
    \end{tikzcd}
  \end{equation}
  But here the left diagonal map is null-homotopic, as indicated, already by degree reasons. This shows that $[W_{a/b}, \zeta] = \mathrm{e}$. \addedinline{Analogously one finds that the remaining group commutators all vanish,} which completes the proof.
\end{proof}

This result generalizes to the other connected components by the following argument adapted from \cite[Prop. 1]{Hansen1974}:
\begin{lemma}
\label[lemma]
 {ConnectedComponentsOfMappingSpaceAreEquiv}
  The connected components of 
  $\PointedMaps\bracket({
      (S^{2k-1})^2,
      S^{2k}
    })$ are all homotopy equivalent.
\end{lemma}
\begin{proof}
  Consider the H-action of the H-group $ \PointedMaps(S^{4k-2}, S^{2k})$ (\cref{HGroupStructureOnMappingSpace}) on the mapping space, which is induced by the pinching map (from \cref{PinchingHCoaction}):
  \begin{equation}
    \begin{tikzcd}[
     column sep=18pt
    ]
      \PointedMaps\bracket({
        S^{4k-2}
        , 
        S^{2k}
      })
      \times
      \PointedMaps\bracket({
        (S^{2k-1})^2
        ,\,
        S^{2k}
      })
      \ar[
        r,
        "{ 
          \Theta 
        }"
      ]
      &
      \PointedMaps\bracket({
        (S^{2k-1})^2
        ,\,
        S^{2k}
      }).
    \end{tikzcd}
  \end{equation}
  Being an H-group action, each $\Theta(f,-)$ is a homotopy auto-equivalence of the mapping space. But then for $[f]$ a generator of $\pi_{4k-2}(S^{2k})$, \cref{OnConnectedComponentsOfTheMappingSpace} shows that $\Theta(f,-)$ restricts to a sequence of homotopy equivalences between all its connected components.
\end{proof}

Finally, we adapt these results to the unpointed mapping space:
\begin{theorem}
\label[theorem]{FinalTheorem}
  The fundamental groups of the unpointed mapping spaces from $(S^{2k-1})^2$ to $S^{2k}$ are integer Heisenberg groups at level $\ell = 2$ \textup{(\cref{IntegerHeisenbergGroupAtLevel2})}, as follows:
  \begin{subequations}
  \begin{align}
    \label{EndResultFork=1}
    \pi_1\,
    \MapsComponent{n}\bracket({
      (S^1)^2
      ,\,
      S^2
    })
    & \simeq
    \mathrm{Heis}_3(\mathbb{Z}, 2n)
    \mathrlap{\,,}
    \\
    \label{EndResultFork=2}
    \pi_1\,
    \MapsComponent{[n]}\bracket({
      (S^3)^2
      ,\,
      S^4
    })
    & \simeq
    \mathrm{Heis}_3(\mathbb{Z}, 0)
    \times^{\addedinline{\prime}}
    \CyclicGroup{12}
    \mathrlap{\,,}
    \\
    \label{EndResultFork=4}
    \pi_1\,
    \MapsComponent{[n]}\bracket({
      (S^7)^2
      ,\,
      S^8
    })
    & \simeq
    \mathrm{Heis}_3(\mathbb{Z}, 0)
    \times^{\addedinline{\prime}}
    \CyclicGroup{120}
    \mathrlap{\,,}
  \end{align}
  \end{subequations}
  for all $n \in \mathbb{Z}$ and $[n] \in \CyclicGroup{2}$, respectively \cref{ConnectedComponentsOfMappingSpace},
  \addedinline{
    with the prime again denoting the central product \cref{TheCentralProduct}}.
\end{theorem}
\begin{proof}
  For $k \geq 2$ the homotopy long exact sequence of the basepoint evaluation fibration truncates: 
  \begin{equation}
    \begin{tikzcd}[
      row sep=0pt,
      column sep=16pt
    ]
      \PointedMapsComponent{[n]}\bracket({%
        (S^{2k-1})^2%
        ,\,%
        S^{2k}
      })
      \ar[
        r,
        "{
          \mathrm{fib}_{\ast}
        }"
      ]
      &
      \MapsComponent{[n]}\bracket({%
        (S^{2k-1})^2%
        ,\,%
        S^{2k}
      })
      \ar[
        r,
        "{ \mathrm{ev}_{\ast} }"
      ]
      &
      S^{2k}
      \\[-6pt]
      &&
      \grayoverbrace{
        \pi_2 S^{2k}
      }{ 1 }
      \ar[
        dll,
        snake left
      ]
      \\[+6pt]
      \pi_1
      \PointedMapsComponent{[n]}\bracket({%
        (S^{2k-1})^2%
        ,\,%
        S^{2k}
      })
      \ar[
        r,
        "{ \sim }"
      ]
      &
      \pi_1
      \MapsComponent{[n]}\bracket({%
        (S^{2k-1})^2%
        ,\,%
        S^{2k}
      })
      \ar[
        r,
        "{ \mathrm{ev}_{\ast} }"
      ]
      &
      \grayunderbrace{
      \pi_1 S^{2k}}{1}
      \mathrlap{\,,}
    \end{tikzcd}
  \end{equation}
  whence the claim for $k \geq 2$ follows immediately by \cref{FundamentalGroupOfMapsFromTorusToSphere,ConnectedComponentsOfMappingSpaceAreEquiv}. \\
  In the case $k = 1$, the presence of $\pi_2(S^2) \simeq \mathbb{Z}$ obstructs this direct argument, and our \cref{FundamentalGroupOfMapsFromTorusToSphere,ConnectedComponentsOfMappingSpaceAreEquiv} gives \cref{EndResultFork=1} without determining the periodicity to be $2n$.
  But this value is fixed by \cref{TheShortExactSequence} below.
\end{proof}

\begin{remark}
  Readers familiar with \emph{rational homotopy theory} (cf. \parencites{FHT2000}[\S III]{FSS23-Char}) may find it instructive to reproduce the statement of \cref{FinalTheorem} at the level of \emph{minimal Sullivan models}:
  For a $X$ connected nilpotent space of finite rational type, we denote its minimal Sullivan model (over some rational ground field $\mathbb{K}$) by: 
  \begin{equation}
    \mathrm{CE}\bracket({\mathfrak{l}X}) 
      \in 
    \mathrm{dgCAlg}_{\mathbb{\mathbb{K}}}
  \end{equation}
  (being the \emph{Chevalley-Eilenberg algebra} $\mathrm{CE}(-)$ of the \emph{rational Whitehead bracket $L_\infty$-algebra} $\mathfrak{l}X$ of $X$, cf. \cite[Prop. 5.11 \& Rem. 5.4]{FSS23-Char}).\\
  The minimal models of positive even-dimensional spheres are (cf. \cite[\S 1.2]{Menichi2015})%
  \footnote{
    Here $\mathbb{K}_{{}_{\mathrm{d}}}\![L]$ denotes the free differential graded-commutative algebra on a list $L$ of graded generators, and the quotient divides out the shown differential ideal. 
  }
  $$
    \mathrm{CE}\bracket({\mathfrak{l}S^{2k}})
    \simeq
    \mathbb{\mathbb{K}}_{{}_{\mathrm{d}}}
    \!\!
    \left[{\!
      \def\arraystretch{1.3}
      \begin{array}{l}
        f
        \\
        h
      \end{array}
    \!}\right]
    \!\Big/\!
    \left({
    \begin{aligned}
      \mathrm{d}\, f      & = 0
      \\
      \mathrm{d}\, h
      & = 
      \tfrac{1}{2}
      f \wedge f
    \end{aligned}
    }\right)
    \,,
    \;\;\;
    \begin{aligned}
      \mathrm{deg}(f)
      & = 2k
      \\
      \mathrm{deg}(h)
      & = 4k-1
      \mathrlap{\,,}
    \end{aligned}
  $$
  and the generators of the minimal model for the mapping space out of a squared odd-dimensional sphere are (cf. \emph{toroidification} in \parencites[p. 9]{SatiVoronov2025}[p. 35]{GSS25-TD}) given by: 
  $$
    \mathrm{CE}\bracket({
      \mathfrak{l}
      \MapsComponent{0}\bracket({
        (S^{2k-1})^2
        ,
        S^{2k}
      })
    })
    \simeq
    \mathbb{K}_{_{\mathrm{d}}}
    \!\!
    \left[{
    \def\arraystretch{1.25}
    \begin{array}{l}
      f
      \\
      h
      \\
      \overset{a}{\mathrm{s}}
      f
      \\
      \overset{b}{\mathrm{s}}
      f
      \\
      \overset{a}{\mathrm{s}}
      h
      \\
      \overset{b}{\mathrm{s}}
      h
      \\
      \overset{a}{\mathrm{s}}
      \overset{b}{\mathrm{s}}
      h
    \end{array}%
    \!\!}\right]
    \!\bigg/\!
    \left(
    \begin{aligned}
      \mathrm{d}\, f 
        & = 0
      \\[-2pt]
      \mathrm{d}\, h 
        & = 
        + \tfrac{1}{2}
        f \wedge f
      \\[-2pt]
      \mathrm{d}\, 
      \overset{a}{\mathrm{s}}
      f
      &=
      0
      \\[-2pt]
      \mathrm{d}\, 
      \overset{b}{\mathrm{s}}
      f
      & = 
      0
      \\[-2pt]
      \mathrm{d}\, 
      \overset{a}{\mathrm{s}}
      h
      & =
      - 
      f 
        \wedge 
      \overset{a}{\mathrm{s}}f
      \\[-2pt]
      \mathrm{d}\, 
      \overset{b}{\mathrm{s}}
      h
      & =
      -
      f 
        \wedge 
      \overset{b}{\mathrm{s}}f
      \\[-2pt]
      \mathrm{d}\, 
      \overset{a}{\mathrm{s}}
      \overset{b}{\mathrm{s}}
      h
      & =
      +
      \overset{a}{\mathrm{s}}f
      \wedge
      \overset{b}{\mathrm{s}}f
    \end{aligned}
    \right)
    ,
    \begin{array}{l}
    \mathrm{deg}\bracket({
      \overset{\bullet}{\mathrm{s}}
      (\text{-})
    })
    =
    \\
    \;\mathrm{deg}({\text{-}})
    -
    2k + 1.
    \end{array}
  $$
  These co-binary differentials dually exhibit the Whitehead brackets over $\mathbb{K}$ (cf. \parencites[Thm. 6.1]{AndrewsArkowitz1978}[Prop. 13.16]{FHT2000}).
  In particular, the last line dually is the Whitehead bracket shown above in \cref{FirstDiagramForDerivingHeisenberg}, now over $\mathbb{K}$. For $\mathbb{K} \defneq \mathbb{R}$, this is the nontrivial Lie bracket of the Lie algebra of the ordinary Heisenberg group (\cref{OrdinaryHeisenbergGroup}). Of course, the crucial level $\ell = 2$ of this Whitehead bracket, which we established above, cannot be discerned in rational homotopy theory (where all non-vanishing factors correspond to isomorphic models).  
\end{remark}

\subsection{Planck's constant}

For completeness, we recall in streamlined form the argument (following \cite{Hansen1974}) for the order $2n$ of the integer Heisenberg group extension (called \emph{Planck's constant} in \cref{PlancConstantAndCSLevel}) in the case $k = 1$, \cref{EndResultFork=1} in \cref{FinalTheorem}.

\begin{lemma}
\label[lemma]
  {MappingSpacesOutOfWedgeSumOfCirclesIntoS2}
  The following forgetful map \textup{(from the fundamental groups of the pointed to that of the unpointed mapping space)} is an isomorphism:
  $$
    \begin{tikzcd}
      \underbrace{
      \pi_1
      \PointedMaps\bracket({
        S^1 \vee S^1
        ,
        S^2
      })
      }_{ \mathbb{Z}^2 }
      \ar[r, "{ \sim }"]
      &
      \pi_1
      \Maps\bracket({
        S^1 \vee S^1
        ,
        S^2
      })
      \mathrlap{\,.}
    \end{tikzcd}
  $$
\end{lemma}
\begin{proof}
  The evaluation fiber sequence has a section (landing in the single connected component of $\mathrm{Map}(S^1 \vee S^1, S^2)$),
  whence the connecting homomorphism vanishes,
  \[
    \begin{tikzcd}[
      column sep=20pt,
      row sep=10pt
    ]
     &
     \phantom{---------}
     \ar[r, ->>]
     &
     \pi_2 S^2
     \ar[
       dll,
       snake left,
       "{ 0 }"{description}
     ]
     \\
     \pi_1\PointedMaps\bracket({
       S^1 \vee S^1
       ,
       S^2
     })
     \ar[r, "{ \sim }"]
     &
     \pi_1\Maps\bracket({
       S^1 \vee S^1
       ,
       S^2
     })
     \ar[r]
     &
     \underbrace{
       \pi_1 S^2
     }_{1}
     \mathrlap{\,,}
    \end{tikzcd}
  \]
  because the previous map is, therefore, split surjective.
\end{proof}

\begin{lemma}
\label[lemma]
{ScndHmtpOfMapS1VeeeS1ToS2}
  We have
  \[
    \pi_2
    \Maps
    \bracket({
      S^1 \vee S^1,
      S^2
    })
    \simeq
    \mathbb{Z}^3
    \mathrlap{\,.}
  \]
\end{lemma}
\begin{proof}
Analogous to the proof of \cref{MappingSpacesOutOfWedgeSumOfCirclesIntoS2} we obtain the LES
  \begin{equation}
  \label
  {LESForScndHmtpOfMapS1VeeeS1ToS2}
    \begin{tikzcd}[
      column sep=20pt,
      row sep=10pt
    ]
     &
     \phantom{---------}
     \ar[r, ->>]
     &
     \pi_3 S^2
     \ar[
       dll,
       snake left,
       "{ 0 }"{description}
     ]
     \\
     \underbrace{
     \pi_2\PointedMaps\bracket({
       S^1 \vee S^1
       ,
       S^2
     })}_{ \mathbb{Z}^2 }
     \ar[r]
     &
     \pi_2\Maps\bracket({
       S^1 \vee S^1
       ,
       S^2
     })
     \ar[r]
     &
     \underbrace{
       \pi_2 S^2
     }_{\mathbb{Z}}
     \mathrlap{\,,}
    \end{tikzcd}
  \end{equation}
  which implies the claim.
\end{proof}

\begin{lemma}[{\parencites[Thm. 5.3(i)]{STHu1946}[Lem. 3.9]{Koh1960}}]
\label[lemma]{FundamentalGroupsOfMapsS2ToS2}
  For all $n \in \mathbb{Z}$, we have
  $$
    \pi_1 \MapsComponent{n}\bracket({
      S^2,
      S^2
    })
    \simeq
    \CyclicGroup{2n}
    \,.
  $$
\end{lemma}
\begin{proof}
  The evaluation homotopy fiber sequence yields the following homotopy long exact sequence:
  \begin{equation}
  \label
  {TowardsFundamentalGroupsOfMapsS2ToS2}
    \begin{tikzcd}[
      column sep=20pt,
      row sep=0pt
    ]
      \PointedMapsComponent{n}\bracket({
        S^2
        ,
        S^2
      })
      \ar[r]
      &
      \MapsComponent{n}\bracket({
        S^2
        ,
        S^2
      })
      \ar[
        r, 
        "{ \mathrm{ev}_\ast }"
      ]
      &
      S^2
      \\[-5pt]
      &&
      \overbrace{
        \pi_2(S^2)
      }^{\mathbb{Z}}
      \ar[
        dll,
        snake left,
        "{ 
          \bracket[{
           [n \cdot \mathrm{id}_2]
           ,
           -
          }]_{\mathrm{Wh}} 
        }"{swap}
      ]
      \\[+10pt]
      \underbrace{
      \pi_1
      \PointedMapsComponent{n}\bracket({
        S^2
        ,
        S^2
      })
      }_{
        \pi_3(S^2) \,\simeq\, \mathbb{Z}
      }
      \ar[r]
      &
      \pi_1
      \MapsComponent{n}\bracket({
        S^2
        ,
        S^2
      })
      \ar[r]
      &
      \underbrace{
      \pi_1(S^2)
      }_{1}
      \mathrlap{\,,}
    \end{tikzcd}
  \end{equation}
  where the connecting homomorphism is given by the Whitehead bracket with $n \cdot \mathrm{id}_2$, by \cite[(3.4)${}^\ast$]{Whitehead1946}. But,
  due to linearity of the Whitehead bracket and by \cref{WhiteheadSquareRootsOfHopfFibrations}, this means that the connecting homomorphism is given by multiplication with $2n$:
  \[
    \begin{aligned}
    1 = [\mathrm{id}_2]
    \;\; \longmapsto  \;\; 
    \bracket[{
      \bracket[{n\cdot \mathrm{id}_2}]
      ,
      \bracket[{\mathrm{id}_2}]
    }]_{\mathrm{Wh}}
     &= 
    n \cdot \bracket[{
      \bracket[{\mathrm{id}_2}]
      ,
      \bracket[{\mathrm{id}_2}]
    }]_{\mathrm{Wh}}
    \\
    & = 
    n \cdot [2 h_{\mathbb{C}}]
    \\
    & = 
    2n \;\; \in \mathbb{Z} \simeq \pi_3(S^2)
    \mathrlap{\,.}
    \end{aligned}
  \]
  From this, the claim follows by exactness of the above sequence.
\end{proof}
We now obtain a short and transparent new proof of:
\begin{proposition}[{\cite[Thm. 1]{Hansen1974}}]
\label[proposition]{TheShortExactSequence}
  For $n \in \mathbb{Z}$, we have a short exact sequence of groups of this form:
  $$
    \begin{tikzcd}[
      column sep=10pt
    ]
      1 
      \ar[r]
      &
      \CyclicGroup{2n}
      \ar[rr]
      &&
      \pi_1 \MapsComponent{n}\bracket({
        (S^1)^2
        ,\,
        S^2
      })
      \ar[rr]
      &&
      \mathbb{Z}^2
      \ar[r]
      &
      1
      \mathrlap{\,.}
    \end{tikzcd}
  $$
\end{proposition}
\begin{proof}
First note that the canonical CW-complex structure on $(S^1)^2$ gives rise to a homotopy cofiber sequence of the form
\begin{equation}
  \label{TorusCofiberSequence}
  \begin{tikzcd}[sep=17pt]
    S^1
    \ar[
      rr,
      "{ [\alpha,\,\beta] }"
    ]
    &&
    S^1 \vee S^1
    \ar[r]
    &
    (S^1)^2
    \ar[r]
    &
    S^2
    \mathrlap{,}
  \end{tikzcd}
\end{equation}
where $\alpha$ and $\beta$ denote a pair of representatives of a pair of generators of $\pi_1(S^1 \vee S^1)$, so that their group commutator $[\alpha,\beta]$ is the \emph{fundamental polygon} of the torus (as seen in \cref{HeisenbergGroupCommutatorIllustrated}).

Then observe that the corresponding homotopy fiber sequence of \emph{unpointed} mapping spaces is 
\[
  \begin{tikzcd}
    \PointedMaps\bracket({
      S^2,
      S^2
    })
    \ar[r]
    &
    \Maps\bracket({
      \bracket({S^1})^2,
      S^2
    })
    \ar[r]
    &
    \Maps\bracket({
      S^1 \vee S^1,
      S^2
    })
  \end{tikzcd}
\]
because $\Maps(-,S^2)$ sends colimits to limits, and using the pasting law:
\[
  \begin{tikzcd}
    \PointedMaps\bracket({
      S^2,
      S^2
    })
    \ar[r]
    \ar[d]
    \ar[
      dr,
      phantom,
      "{ \lrcorner }"{pos=.1}
    ]
    &
    \Maps\bracket({
      S^2,
      S^2
    })
    \ar[r]
    \ar[d]
    \ar[
      dr,
      phantom,
      "{ \lrcorner }"{pos=.1}
    ]
    &
    \Maps\bracket({
      \bracket({S^1})^2,
      S^2
    })
    \ar[d]
    \\
    \ast
    \ar[r]
    &
    \Maps\bracket({
      \ast, 
      S^2
    })
    \ar[r]
    &
    \Maps\bracket({
      S^1 \vee S^1,
      S^2
    })
    \mathrlap{\,.}
  \end{tikzcd}
\]

Consider then the corresponding homotopy long exact sequence 
\[
  \begin{tikzcd}[
    column sep=15pt,
    row sep=5pt
  ]
    &&
    \overbrace{
    \pi_2\Maps\bracket({
      S^1 \vee S^1
      ,
      S^2
    })}^{
      \mathbb{Z}^3
    }
    \ar[
      dll,
      snake left,
      "{
        \partial_1
      }"{description}
    ]
    \\
    \underbrace{
    \pi_1\PointedMapsComponent{n}\bracket({
      S^2
      ,
      S^2
    })
    }_{ \mathbb{Z} }
    \ar[r]
    &
    \pi_1\MapsComponent{n}\bracket({
      (S^1)^2
      ,
      S^2
    })
    \ar[r]
    &
    \underbrace{
    \pi_1\Maps\bracket({
      S^1 \vee S^1
      ,
      S^2
    })
    }_{ \mathbb{Z}^2 }
    \ar[
      dll,
      snake left
    ]
    \\[+7pt]
    \pi_0\PointedMapsComponent{n}\bracket({
      S^2
      ,
      S^2
    })
    \mathrlap{\,,}
  \end{tikzcd}
\]
where the braced terms on the right are evaluated by \cref{MappingSpacesOutOfWedgeSumOfCirclesIntoS2,ScndHmtpOfMapS1VeeeS1ToS2}.
Now observe that the connecting homomorphism is
\[
  \partial_1
  =
  \bracket({
    (0,0), 2n
  })
  :
  \begin{tikzcd}
    \mathbb{Z}^2 \oplus \mathbb{Z}
    \ar[r]
    &
    \mathbb{Z}
    \mathrlap{\,.}
  \end{tikzcd}
\]
Namely, using the decomposition $\mathbb{Z}^2 \oplus \mathbb{Z}$ from \cref{LESForScndHmtpOfMapS1VeeeS1ToS2}, it is zero on the $\mathbb{Z}^2$ summand, since there it pulls back the group commutators to an abelian group (this is the key move in \cite[Prop. 2]{Hansen1974}), and it multiplies by $2n$ on the remaining $\mathbb{Z}$-summand, by \cref{TowardsFundamentalGroupsOfMapsS2ToS2}.
By exactness, this implies the claim.
\end{proof}

\section{Generalizations}
\label{OnGeneralizations}

We may generalize the situation in \cref{OnTheTheorem} from the case of ``squared spheres'' $(S^{2k-1})^2$ to more general CW-complexes $M$ (such as higher genus surfaces, \cref{HigherGenusSurface}), where the Whitehead product attaching map \cref{CWComplexStructureOnSquaredSphere} is allowed to have more arguments and to have multiplicities (\cref{MoreGeneralCWComplex} below). With \cref{SmashOfProjectionMaps} below established, the proofs in this generality follow the same logic as those in \cref{OnTheTheorem}, whence we may be more brief.

\begin{definition}
\label[definition]{MoreGeneralCWComplex}
For $k \in \mathbb{N}_{\geq 1}$ and $i,j$ ranging through some finite set $I$,
let $M$ be the CW-complex given by the following pushout square on the left
\begin{equation}
  \label{GeneralizedCWPushout}
  \begin{tikzcd}[row sep=12pt, column sep=large]
    S^{4k-3}
    \ar[
      r,
      "{
        \phi
      }"
    ]
    \ar[d]
    \ar[
      dr,
      phantom,
      "{ \ulcorner }"{pos=.9}
    ]
    &
    \bigvee_{i}
    S^{2k-1}
    \ar[d]
    \ar[r]
    \ar[
      dr,
      phantom,
      "{ \ulcorner }"{pos=.9}
    ]
    &
    \ast
    \ar[d]
    \\
    D^{4k-2}
    \ar[r]
    &
    M
    \ar[
      r,
      "{ 
        \mathrm{pr}_c 
      }"{swap}
    ]
    \ar[
      d,
      "{ 
        \mathrm{pr}_j 
    }"
    ]
    &
    S^{4k-2}
    \\[+3pt]
    & 
    S^{2k-1}
    \mathrlap{\,,}
  \end{tikzcd}
\end{equation}
where the attaching map is a product
\begin{equation}
  \label{GeneralizedAttachingMap}
  \phi
  :=
  \prod_{i < j}
  \bracket({
    \bracket[{
      \iota_i
      ,
      \iota_j
    }]_{\mathrm{Wh}}
  })^{\lambda_{i j}}
  ,
  \;\;\;
  \lambda_{i j} 
  =
  \addedinline{
  -
  \lambda_{j i}
  }
  \in \mathbb{Z}
  \,,
\end{equation}
of Whitehead products:
$$
  \begin{tikzcd}[column sep=large]
    S^{4k-3}
    \ar[r]
    \ar[d]
    \ar[
      dr,
      phantom,
      "{ \ulcorner }"{pos=.9}
    ]
    \ar[
      rrr,
      uphordown,
      "{
        [
          \iota_i
          ,\,
          \iota_j
        ]_{\mathrm{Wh}}
      }"
    ]
    &
    S^{2k-1} 
      \vee
    S^{2k-1}
    \ar[
      rr, 
      hook,
      "{
        (i_i, i_j)
      }"
    ]
    \ar[d]
    &&
    \bigvee_i S^{2k-1}
    \\
    D^{4k-2}
    \ar[r]
    &
    S^{2k-1}
      \times
    S^{2k-1}
    \mathrlap{\,,}
  \end{tikzcd}
$$
and where the map $\mathrm{pr}_j$ is a chosen extension to $M$ of the projection of the $2k-1$-skeleton $\bigvee_i S^{2k-1}$ onto the $j$th sphere factor. Such an extension is possible due to the fact that $\phi$ \cref{GeneralizedAttachingMap} becomes null-homotopic under this projection:
\begin{equation}
  \label{ProjectionOnCellOfSkeleton}
  \begin{tikzcd}[row sep=small, column sep=large]
    S^{4k-3}
    \ar[
      r,
      "{
        \phi
      }"
    ]
    \ar[d]
    \ar[
      dr,
      phantom,
      "{ \ulcorner }"{pos=.9}
    ]
    &
    \bigvee_{i}
    S^{2k-1}
    \ar[d]
    \ar[
      ddr,
      bend left=10pt,
      "{ 
        (\delta_i^j)_{i \in I}
      }"
    ]
    \\
    D^{4k-2}
    \ar[r]
    \ar[
      drr,
      bend right=10pt,
      "{ \ast }"{swap}
    ]
    &
    M
    \ar[
      dr,
      "{
        \mathrm{pr}_j
      }"{description, pos=.45}
    ]
    \\
    && 
    S^{2k-1}
    \mathrlap{.}
  \end{tikzcd}
\end{equation}
\end{definition}

\begin{SCfigure}[1][htbp]
\caption{
  \label{FundamentalPolygons}
  Oriented surfaces arise by identifying boundary segments of their \emph{fundamental polygons} such that their 2-cell attaching map is an iterated Whitehead product \cref{2CellAttachmentForClosedSurfaces}.
}
\adjustbox{
  raise=-.3cm
}{
\adjustbox{
  scale=.8,
  raise=-1.55cm
}{
\begin{tikzpicture}
  \node  at (1,2.7)
  {
    \scalebox{.9}{
      \color{darkblue}
      \bf
      sphere
    }
  };
  \draw[
    line width=3,
    fill=lightgray
  ]
    (0,0)
    rectangle 
    (2,2);

\draw[fill=black]
  (0,0) circle (.08);

\draw[fill=black]
  (2,0) circle (.08);

\draw[fill=black]
  (0,2) circle (.08);

\draw[fill=black]
  (2,2) circle (.08);

\node
  at (1,.9) {$
    S^2_0 \simeq S^2
  $};

\end{tikzpicture}
}
\;\;\;
\adjustbox{
  scale=.8,
  raise=-1.6cm
}{
\begin{tikzpicture}
  \node  at (1,2.7)
  {
    \scalebox{.9}{
      \color{darkblue}
      \bf
      torus
    }
  };
  \draw[
    draw opacity=0,
    fill=lightgray
  ]
    (0,0)
    rectangle 
    (2,2);

\draw[
  dashed,
  color=darkgreen,
  line width=1.6,
]
  (0,2) -- 
  node[yshift=8pt]{
    \scalebox{1}{
      \color{black}
      $a$
    }
  }
  (2,2);
\draw[
  -Latex,
  darkgreen,
  line width=1.6
]
  (1.3-.01,2) -- 
  (1.3,2);

\draw[
  dashed,
  color=darkgreen,
  line width=1.6,
]
  (0,0) -- (2,0);
\draw[
  -Latex,
  darkgreen,
  line width=1.6
]
  (1.3-.01,0) -- 
  (1.3,0);

\draw[
  dashed,
  color=olive,
  line width=1.6,
]
  (0,0) -- (0,2);
\draw[
  -Latex,
  olive,
  line width=1.6
]
  (0, .8+.01) -- 
  (0,.8);

\draw[
  dashed,
  color=olive,
  line width=1.6,
]
  (2,0) -- 
  node[xshift=8pt]{
    \scalebox{1}{
      \color{black}
      $b$
    }
  }
  (2,2);
\draw[
  -Latex,
  olive,
  line width=1.6
]
  (2, .8+.01) -- 
  (2,.8);

\draw[fill=black]
  (0,0) circle (.08);

\draw[fill=black]
  (2,0) circle (.08);

\draw[fill=black]
  (0,2) circle (.08);

\draw[fill=black]
  (2,2) circle (.08);

\node
  at (1,.9)
  {
    $S^2_1 \simeq \mathbb{T}^2$
  };

\end{tikzpicture}
}
\hspace{-23pt}
\adjustbox{
  raise=-2.3cm
}{
  \begin{tikzpicture}[
    scale=.65
  ]

  \node[
    rotate=39
  ] at (-1.8,1.8) {
    \scalebox{.64}{
      \color{darkblue}
      \bf
      \def\arraystretch{.7}
      \begin{tabular}{c}
        2-holed
        \\
        torus
      \end{tabular}
    }
  };

  \draw[
    draw opacity=0,
    fill=lightgray
  ]
    (90-22.5:2) --
    (45+90-22.5:2) --
    (2*45+90-22.5:2) --
    (3*45+90-22.5:2) --
    (4*45+90-22.5:2) --
    (5*45+90-22.5:2) --
    (6*45+90-22.5:2) --
    (7*45+90-22.5:2) --
    cycle;

    \draw[
      line width=1.6,
      dashed,
      darkgreen
    ]
      (90-22.5:2)
      --
      (90+22.5:2);

    \draw[
      line width=1.6,
      dashed,
      olive
    ]
      (-45+90-22.5:2)
      --
      (-45+90+22.5:2);

    \draw[
      line width=1.6,
      dashed,
      darkgreen
    ]
      (-90+90-22.5:2)
      --
      (-90+90+22.5:2);

    \draw[
      line width=1.6,
      dashed,
      olive
    ]
      (-135+90-22.5:2)
      --
      (-135+90+22.5:2);

    \draw[
      line width=1.6,
      dashed,
      darkblue
    ]
      (-180+90-22.5:2)
      --
      (-180+90+22.5:2);

    \draw[
      line width=1.6,
      dashed,
      darkblue
    ]
      (-180+90-22.5:2)
      --
      (-180+90+22.5:2);

    \draw[
      line width=1.6,
      dashed,
      purple
    ]
      (-225+90-22.5:2)
      --
      (-225+90+22.5:2);

    \draw[
      line width=1.6,
      dashed,
      darkblue
    ]
      (-270+90-22.5:2)
      --
      (-270+90+22.5:2);

    \draw[
      line width=1.6,
      dashed,
      purple
    ]
      (-315+90-22.5:2)
      --
      (-315+90+22.5:2);

\draw[
  line width=1.3,
  -Latex,
  darkgreen
]
  (0.3,1.83) -- (.35,1.83);

\begin{scope}[
  rotate=-45
]
\draw[
  line width=1.3,
  -Latex,
  olive
]
  (0.24,1.83) -- 
  (.25,1.83);
\end{scope}
\begin{scope}[
  xscale=-1,
  rotate=+90
]
\draw[
  line width=1.3,
  -Latex,
  darkgreen
]
  (0.25,1.83) -- 
  (.3,1.83);
\end{scope}
\begin{scope}[
  xscale=-1,
  rotate=+135
]
\draw[
  line width=1.3,
  -Latex,
  olive
]
  (0.25,1.83) -- 
  (.3,1.83);
\end{scope}

\begin{scope}[
  rotate=+180
]
\draw[
  line width=1.3,
  -Latex,
  darkblue
]
  (.23,1.83) -- 
  (.33,1.83);
\end{scope}

\begin{scope}[
  rotate=-225
]
\draw[
  line width=1.3,
  -Latex,
  purple
]
  (.2,1.83) -- 
  (.3,1.83);
\end{scope}

\begin{scope}[
  rotate=-270
]
\draw[
  line width=1.3,
  -Latex,
  darkblue
]
  (-.20,1.83) -- 
  (-.25,1.83);
\end{scope}
\begin{scope}[
  rotate=-315
]
\draw[
  line width=1.3,
  -Latex,
  purple
]
  (-.20,1.83) -- 
  (-.25,1.83);
\end{scope}

\foreach \n in {1,...,8} {
  \draw[
    fill=black
  ]
    (22.5+\n*45:2)
    circle
    (.11);
};

\node[
  scale=.9
] at
  (0,2.3) {
    $a_1$
  };

\begin{scope}[
  rotate=-45
]
  \node[scale=.9] at
    (0,2.3) {
      $b_1$
    };
\end{scope}

\begin{scope}[
  rotate=-180
]
  \node[scale=.9] at
    (0,2.3) {
      $a_2$
    };
\end{scope}

\begin{scope}[
  rotate=-180-45
]
  \node[scale=.9] at
    (0,2.3) {
      $b_2$
    };
\end{scope}

\node
  at (0,0) {
    $S^2_2$
  };
  
\end{tikzpicture}
}
}
\end{SCfigure}
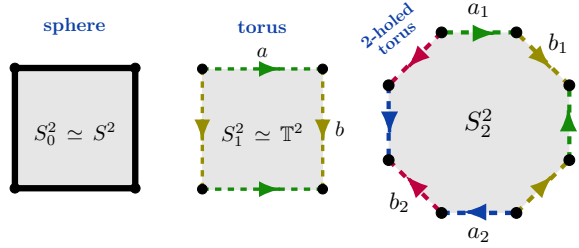

\begin{example}
\label[example]{HigherGenusSurface}
  The closed oriented \emph{surface} $S^2_g$ of genus $g \in \mathbb{N}$ is an instance of \cref{MoreGeneralCWComplex} for $k = 1$, where the pushout reflects the classical construction of $S^2_g$ by identifications among the boundary segments of a \emph{fundamental polygon} (cf. \cite[Thm. 2.8]{Giblin1977} and \cref{FundamentalPolygons}):
  \begin{equation}
    \label{2CellAttachmentForClosedSurfaces}
    \begin{tikzcd}[
      column sep=40pt
    ]
      S^1
      \ar[
        rr,
        "{
          [
            \iota_1
            ,
            \iota_{g+1}
          ]
          \cdots
          [
            \iota_g
            ,
            \iota_{2g}
          ]
        }"
      ]
      \ar[d]
      &
      {}
      \ar[
        dr,
        phantom,
        "{ \ulcorner }"{pos=.9}
      ]
      &
      \bigvee^{2g}
      S^1
      \ar[d]
      \\
      D^2
      \ar[rr]
      &&
      S^2_g
      \mathrlap{\,.}
    \end{tikzcd}
  \end{equation}
\end{example}
The key to generalizing the theorem of \cref{OnTheTheorem} is now the following observation, generalizing \cref{IdentityBetweenProjections}:
\begin{lemma}

\label[lemma]{SmashOfProjectionMaps}
With $M$ as in \cref{MoreGeneralCWComplex}, 
the following diagrams commute up to homotopy, for $j < j' \in I$:
\begin{equation}
  \begin{tikzcd}[
    row sep=12pt, 
    column sep=40pt
  ]
    &
    M \wedge M
    \ar[
      dr,
      "{
        \mathrm{pr}_j
        \wedge
        \mathrm{pr}_{j'}
      }"{sloped}
    ]
    \\
    M
    \ar[
      ur,
      "{ \Delta }"{sloped}
    ]
    \ar[
      rr,
      "{
        \lambda_{j j'}
        \cdot
        \mathrm{pr}_c
      }"
    ]
    &&
    S^{4k-2}
    \mathrlap{.}
  \end{tikzcd}
\end{equation}
\end{lemma}
\begin{proof}
By the \emph{Hopf-Whitney theorem} (cf. \cite[(6.19) on p. 244]{Whitehead1978}), it is sufficient to show that both the top and the bottom map in this diagram have the same degree, hence that the degree of the top composite is $\lambda_{j j'}$, in that the pullback of the volume class $1 \in \mathbb{Z} \simeq H^{4k-2}(S^{4k-2}; \mathbb{Z})$ is $\lambda_{j j'}$ times the generator of $H^{4k-2}(M; \mathbb{Z})$.

To that end, consider the following diagram:
$$
  \begin{tikzcd}[
    column sep=25pt
  ]
    S^{4k-3}
    \ar[
      d,
      "{
        \phi
      }"{swap, pos=.4}
    ]
    \ar[
      rrr,
      "{
        \lambda_{j j'}
      }"
    ]
    \ar[
      drrr,
      "{
        \lambda_{j j'}
        \cdot
        [\iota,\iota]_{\mathrm{Wh}}
      }"{sloped, description}
    ]
    & &&
    S^{4k-3}
    \ar[
      d,
      "{ [\iota,\iota]_{\mathrm{Wh}} }"
    ]
    \\
    \bigvee_i S^{2k-1}
    \ar[d]
    \ar[
      rrr,
      "{
        p_{j j'}
      }"{description}
    ]
    &&&
    S^{2k-1} \vee S^{2k-1}
    \ar[d]
    \\
    M
    \ar[
      r
    ]
    \ar[
      dr,
      "{
        \Delta
      }"{yshift=-1pt}
    ]
    \ar[dd]
    &
    M \times M
    \ar[
     rr,
     "{
       \mathrm{pr}_j
       \times
       \mathrm{pr}_{j'}
     }"
    ]
    \ar[
      d,
      shorten = -2pt
    ]
    &&
    S^{2k-1}
    \times
    S^{2k-1}
    \ar[dd]
    \\[-14pt]
    & 
    M \wedge M
    \ar[
      drr,
      "{
        \mathrm{pr}_{j}
        \wedge
        \mathrm{pr}_{j'}
      }"{sloped}
    ]
    \\[-10pt]
    S^{4k-2}
    \ar[
      rrr,
      "{
        \lambda_{j j'}
      }"
    ]
    &
    &&
    S^{4k-2}
    \mathrlap{.}
  \end{tikzcd}
$$
Here:
\begin{enumerate}
  \item
  The vertical parts are the homotopy cofiber sequences (``Puppe sequences'').
  \item The map $p_{j j'}$ denotes the projection onto the $j$ and $j'$th wedge summand, whence the top square commutes by definition of $\phi$ \cref{GeneralizedAttachingMap}.
  \item The middle square commutes by \cref{ProjectionOnCellOfSkeleton}.

  \item By functoriality of the pushout, the bottom map is the suspension of the top map and as such of the same degree $\lambda_{j j'}$, as shown.

\end{enumerate}
But then the top right triangle in the bottom square commutes by definition, whence the bottom left triangle exhibits the degree of the diagonal map as claimed.
\end{proof}

With this in hand, we have an immediate generalization of \cref{FundamentalGroupOfMapsFromTorusToSphere}:
\begin{theorem}
\label[theorem]{FundamentalGroupOfMapsFromCWComplexToSphere}
  With $M$ as in \cref{MoreGeneralCWComplex}, and for $k \in \{1,2,4\}$, 
  we have:
  \begin{equation}
    \label{Pi1OfGeneralizedPointedSpaceInZeroComponent}
    \pi_1\,
    \PointedMapsComponent{0}\bracket({
      M
      ,\,
      S^{2k}
    })
    \simeq
    \bracketmid\langle{
      \bracket({W_i})_{i \in I}
      ,\,
      \zeta
    }{
      \bracket[{W_i, W_j}] 
        = 
      \zeta^{2 \lambda_{i j}}
      ,\,
      \bracket[{W_i, \zeta}] 
        = 
      \mathrm{e}
    }\rangle
    \times
      ^{ \addedinline{\prime} } 
    T
    \mathrlap{\,,}
  \end{equation}
  where $T$ stands for the same torsion groups as in \cref{IdentifyingFundamentalGroupOfMapsFromTorusToSphere},
  \addedinline{
  and the central product now is such that
  \[
    \bracket[{W_i, W_j}] 
      = 
    \zeta^{2 \lambda_{i j}}
    \cdot
    [-1]^{\lambda_{i j}}
    \mathrlap{\,.}
  \]
  }
\end{theorem}
\begin{proof}
Since the suspension of Whitehead products is null (by \cref{SuspensionOfWhiteheadVanishes}), it follows from \cref{GeneralizedCWPushout,GeneralizedAttachingMap} that we have a \emph{stable splitting} in generalization of \cref{StableSplittingOfSuspendedSquaredSphere}:
$$
  \begin{tikzcd}[
    column sep=75pt
  ]
  \Sigma M
  \ar[
    r,
    "{
      \bracket({
        (\Sigma\mathrm{pr}_i)_{i \in I}
        ,
        \Sigma \mathrm{pr}_c
      })
    }",
    "{ \sim }"{swap}
  ]
  &
  \bracket({
    \bigvee_i S^{2k}
  })
  \vee
  S^{4k-1}
  \mathrlap{.}
  \end{tikzcd}
$$
This gives the set of generators as claimed in \cref{Pi1OfGeneralizedPointedSpaceInZeroComponent}. 

To see that their group commutators are as claimed, consider the direct analogue of diagram \cref{FirstDiagramForDerivingHeisenberg}, where now the left inner triangle is given by \cref{SmashOfProjectionMaps}:
\begin{equation}
    \begin{tikzcd}[
      column sep=60pt
    ]
      M
      \ar[
        d,
        "{ \Delta }"{swap, pos=.4}
      ]
      \ar[
        rr,
        "{
          \widetilde{%
            [W_j, W_{j'}]%
          }
        }"
      ]
      \ar[
        dr,
        "{
          \lambda_{j j'}
          \cdot
          \mathrm{pr}_c
        }"{description}
      ]
      &&
      \Omega S^{2k}
      \\
      M
      \wedge
      M
      \ar[
        rr,
        downhorup,
        "{
          \widetilde{W}_j
          \,\wedge\,
          \widetilde{W}_{j'}
        }"{description}
      ]
      \ar[
        r,
        "{
          \mathrm{pr}_j
          \wedge 
          \mathrm{pr}_{j'}
        }"{swap}
      ]
      &
      S^{2k-1} 
        \wedge 
      S^{2k-1}
      \ar[
        r,
        "{
          \widetilde{\mathrm{id}_{2k}}
          \,\wedge\,
          \widetilde{\mathrm{id}_{2k}}
        \;}"{swap}
      ]
      \ar[
        ur,
        "{
          \widetilde{
            2 h_{\mathbb{K}}
            \addedinline{
            -
            \kappa_{\mathbb{K}}
            }
          }
        }"{description}
      ]
      &
      (\Omega S^{2k})
      \wedge
      (\Omega S^{2k})
      \mathrlap{\,.}
      \ar[
        u,
        "{
          [-,-]_{\mathrm{Sam}}
        }"{swap}
      ]
    \end{tikzcd}
\end{equation}
Similarly, diagram \cref{SecondDiagramForDerivingHeisenberg} generalizes immediately, which completes the proof.
\end{proof}

It is furthermore straightforward now to pass from \cref{FundamentalGroupOfMapsFromCWComplexToSphere} to the generalization of \cref{FinalTheorem}:
\begin{theorem}
\label[theorem]{GeneralizedFinalTheorem}
  For $M$ as in \cref{MoreGeneralCWComplex}, we have:
  \begin{subequations}
  \begin{align}
    \label{GeneralizedEndResultFork=1}
    \pi_1\,
    \MapsComponent{n}\bracket({
      M
      ,\,
      S^2
    })
    & \simeq
    \bracketmid\langle{
      \bracket({W_i})_{i \in I}
      ,\,
      \zeta
    }{
      \bracket[{W_i, W_{j}}] 
        = 
      \zeta^{2 \lambda_{i j}}
      ,\,
      \bracket[{W_i, \zeta}] 
        = 
      \mathrm{e}
      \,,
      \zeta^{2n} = \mathrm{e}
    }\rangle
    \mathrlap{\,,}
    \\
    \label{GeneralizedEndResultFork=2}
    \pi_1\,
    \MapsComponent{[n]}\bracket({
      M
      ,\,
      S^4
    })
    & \simeq
    \bracketmid\langle{
      \bracket({W_i})_{i \in I}
      ,\,
      \zeta
    }{
      \bracket[{W_i, W_{j}}] 
        = 
      \zeta^{2 \lambda_{i j}}
      ,\,
      \bracket[{W_i, \zeta}] 
        = 
      \mathrm{e}
    }\rangle
    \times
      ^{ \addedinline{\prime} }
    \CyclicGroup{12}
    \mathrlap{\,,}
    \\
    \label{GeneralizedEndResultFork=4}
    \pi_1\,
    \MapsComponent{[n]}\bracket({
      M
      ,\,
      S^8
    })
    & \simeq
    \bracketmid\langle{
      \bracket({W_i})_{i \in I}
      ,\,
      \zeta
    }{
      \bracket[{W_i, W_{j}}] 
        = 
      \zeta^{2 \lambda_{i j}}
      ,\,
      \bracket[{W_i, \zeta}] 
        = 
      \mathrm{e}
    }\rangle
    \times
      ^{ \addedinline{\prime} }
    \CyclicGroup{120}
    \mathrlap{\,.}
  \end{align}
  \end{subequations}
\end{theorem}
\begin{proof}
  With \cref{FundamentalGroupOfMapsFromCWComplexToSphere} in hand, the remaining argument is verbatim that of the proof of \cref{FinalTheorem}, subject just to substituting $M$ for $(S^{2k-1})^2$ with $k = 1$ \cref{GeneralizedEndResultFork=1}, $k = 2$ \cref{GeneralizedEndResultFork=2} or $k = 4$ \cref{GeneralizedEndResultFork=4}, respectively.
\end{proof}

In particular:
\begin{example}
\label[example]{Pi1OfMapsOutOfSurfaceIntoSphere}
  For $k = 1$ and in the case that $M = S^2_g$ is the closed oriented surface  of genus $g$ (\cref{HigherGenusSurface}), \cref{GeneralizedFinalTheorem} gives the higher-dimensional integer Heisenberg groups \cref{UnderlyingSetOfHigherDimIntegerHeis}:
  \begin{equation}
    \pi_1\,
    \MapsComponent{n}\bracket({
      S^2_g
      ,\,
      S^2
    })
    \simeq
    \mathrm{Heis}_{2g+1}(\mathbb{Z}; 2n)
    \mathrlap{\,.}
  \end{equation}
  This is the generality of the situation originally discussed by \parencites{Hansen1974}[Thm. 1]{LarmoreThomas1980}[Prop. 1.5]{Kallel2001}.
\end{example}

\section{\addedinline{Potential} Applications}
\label{OnApplications}

Here we expand (following up on the indications in \cref{OnRelationToFQHAnyons}) on how the result of \cref{OnTheTheorem} may be understood as saying that the \emph{quantum observables} of (abelian) \emph{anyons} on the torus (\cref{HeisenbergGroupCommutatorIllustrated}) are equivalently the group algebra of the 2-Cohomotopy of the suspended torus, and in fact also of the 4-Cohomotopy of the suspended squared 3-sphere and of the 8-Cohomotopy of the suspended squared 7-sphere.

By itself, this is just a mathematical fact (\cref{FinalTheorem}, \ref{EndResultFork=1}). But we recall here how there is physical significance in \emph{magnetic flux quantized in Cohomotopy} \cite{SS25-Flux} which, if physically realized (``Hypothesis h'' \parencites{SS25-FQH}, recalled below in \cref{OnOrdinaryFHQAnyonsIn2Cohomotopy}) renders this mathematical fact, instead of a coincidence, one instance of a novel algebro-topological effective theory of anyons --- which makes predictions, potentially relevant for contemporary materials research, not captured by previous theories.

This appears to be a remarkable new opportunity for (low-dimensional) algebraic topology and homotopy theory to interact with cutting-edge experimental and industry-relevant research (quantum materials, topological quantum computing), potentially of impact comparable to what had been hoped would happen with topological data analysis (TDA, where methods of algebraic topology and homotopy theory are hoped to identify hidden structure in complex data sets).

Moreover, we explain (in \cref{On2DimensionalAnyonsVia4Cohomotopy} below, following \cite{SS25-Complete}) how our generalized result of \cref{FinalTheorem} establishes that quantum observables of (abelian) FQH anyons (of the kind that have been observed in experiment in recent years) arise as topological observables in a completion of \emph{11-dimensional supergravity} (11D SuGra) by flux quantization in 4-Cohomotopy (``Hypothesis H'', now with a capital ``H''). 

As before, by itself this is just a mathematical fact (\cref{FinalTheorem}, \ref{EndResultFork=2}), surprising as it may sound. It does, however, call to mind the much-discussed ``holographic'' relation between 11D SuGra and condensed matter theory \parencites{Herzog2007}{Zaanen2015}{Hartnoll2018}{nLab:HolographicCMT} and may be understood as a novel instance of a general idea of \emph{gauge/gravity duality} whereby strongly-coupled quantum systems (for which there is little traditional theory available) are thought to be mapped onto more tractable dynamics of \emph{branes} fluctuating in auxiliary higher-dimensional spacetimes. 

While this approach has seen immense activity in the last couple of decades, it is notorious for its dearth of precise definitions and of provable theorems. Our result may open the door to a new form of mathematically substantiated interaction between completed quantum 11D SuGra (``M-theory'' \cite{Duff1999World}) and topological quantum materials (following \parencites{SS23-Defect}{SS25-Seifert}{SS25-Rickles}{SS25-Srni}).

\subsection{Ordinary FQH Anyons via 2-Cohomotopy}
\label{OnOrdinaryFHQAnyonsIn2Cohomotopy}

In view of our theorems in \cref{OnTheTheorem,OnGeneralizations}, here we briefly review the understanding of FQH anyons \parencites{SS25-AbelianAnyons}{SS25-FQH} and of FQAH anyons \parencites{SS25-FQAH}{SS25-Crys} via (magnetic or Berry-)flux quantization in 2-Cohomotopy (further surveyed in \cite{SS25-ISQS29}).

\subsubsection{Ordinary Magnetic Flux}
\label{OnOrdinaryMagneticFlux}

Envision a thin slab of material of the form of a closed oriented surface $S^2_g$ (\cref{HigherGenusSurface}) of some thickness $2\epsilon \in \mathbb{R}_{> 0}$ and penetrated transversally by a magnetic field. The phenomenon of \emph{Dirac charge quantization} (cf. \parencites{Alvarez1985}[\S 16.4e]{Frankel2011}) entails (cf. \parencites[\S 3.1]{SS24-Phase}[Cor. 2.3]{SS25-Complete}) that the topological sectors of solitonic magnetic flux configurations through this material form the set of homotopy classes of maps out of its one-point compactification $(-)_{\cpt}$ (enforcing the vanishing of solitonic flux at infinity) to the classifying space of the (``gauge'') group $\mathrm{U}(1)$:
\begin{equation}
  \label{OrdinaryMagneticFluxSectors}
  \begin{aligned}
    \text{Magnetic flux sectors}
    & 
    =
    \pi_0\,
    \PointedMaps\bracket({
      \bracket({
        S^2_g 
          \times 
        (-\epsilon, +\epsilon)
      })_{\cpt}
      ,
      B \mathrm{U}(1)
    })
    \\
    & \simeq
    \pi_0\,
    \PointedMapsComponent{0}\bracket({
      (S^2_g)_+ \wedge S^1
      ,
      B \mathrm{U}(1)
    })
    \\
    & \simeq
    \pi_1\,
    \MapsComponent{0}\bracket({
      S^2_g
      ,
      B \mathrm{U}(1)
    })
    \mathrlap{\,.}
  \end{aligned}
\end{equation}
This is a fundamental group of a mapping space much as we have been discussing in \cref{OnGeneralizations}, only that here the coefficient is $B \mathrm{U}(1) \simeq \mathbb{C}P^\infty$ instead of $S^2 \simeq \mathbb{C}P^1$. With this ``stable'' coefficient it is immediate to compute the fundamental group to be:
$$
  \begin{aligned}
    \cdots
    & 
    \simeq
    \pi_0\,
    \Maps\bracket({
      S^2_g
      ,
      \Omega B \mathrm{U}(1)
    })
    \\
    & \simeq
    \pi_0\,
    \Maps\bracket({
      S^2_g
      ,
      B \mathbb{Z}
    })
    \\
    & \simeq
    H^1\bracket({
      S^2_g
      ;
      \mathbb{Z}
    })
    \\
    & \simeq
    \mathbb{Z}^{2g}
    \mathrlap{\,.}
  \end{aligned}
$$
This is the abelian base group of which $\mathrm{Heis}_{2g+1}(\mathbb{Z})$ (\cref{IntegerHeisenbergGroupAtLevel2}) is a nonabelian extension.
The topological quantum observables on these magnetic flux sectors form the group algebra of $\mathbb{Z}^2$ (by \parencites[\S 1]{SS24-Obs}[(132)]{SS25-Complete}).

\subsubsection{Anyons as Surplus FQH Flux}
\label{OnSurplusFQHFlux}

Envision next that our slab of material is filled with a very cold electron gas and penetrated by a magnetic field so strong that there are \emph{exactly}
\footnote{
  The theory explaining the FQH effect requires indeed an \emph{exact} integer (in general: rational) factor $K$ of magnetic flux quanta over electrons, such as for the composite fermion model to apply \cite{Jain2007}. But a curious property of experimental realization of the FQHE is its famous robustness against deviations of the magnetic field strength from its critical value, making the FQH effect famously appear in small \emph{plateaux} around the critical field strengths. Remarkably, this is explained by \emph{Anderson localization} due to disorder in real-world materials, which localizes electrons/quasi-particles and thereby decouples them from the FQH liquid (cf. \parencites[p. 23]{Laughlin1981}[p. 879-880]{Stormer99}[\S 2.2.1]{Tong2016}{PuSreejithJain2022}{SauDasSarmaEtAl2025}): In an ideal sample without such stabilizing disorder, the external magnetic field would need to be fine-tuned to achieve the exact filling fraction --- in real materials this fine-tuning is, remarkably, effectuated by Anderson localization.
}
some $K$ magnetic flux quanta per electron. 
This is called a \emph{fractional quantum Hall system} at \emph{filling fraction} $\sfrac{1}{K}$ (cf. \parencites{Stormer99}{Tong2016}). 

On such a backdrop, every surplus magnetic flux quantum appears as the lack of $\sfrac{1}{K}$th of an electron and as such is called a fractional \emph{quasi-hole}. It is these \emph{quasi-holes} that behave as anyons, in that under their movement around each other the quantum state of the entire system picks up a complex \emph{braiding phase} $\zeta$ (cf. \cref{FQHAnyons}).

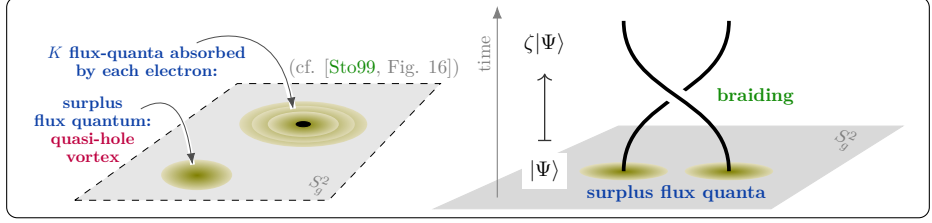
\begin{figure}[htb]
\caption{%
  \label{FQHAnyons}%
  The anyons of fractional quantum Hall systems are vortices in the 2D electon gas induced by surplus magnetic flux quanta on top of an exact rational \emph{filling fraction} of $K$ flux quanta per electron. Under each \emph{braiding} of their worldlines the quantum state $\vert \psi \rangle$ transforms by multiplication with a \emph{braiding phase} $\zeta = \exp(\pi \mathrm{i}/K)$.
}
\centering
\adjustbox{
  rndfbox=4pt
}{
\adjustbox{
  raise=-1.5cm,
  scale=.9,
}{
\hspace{-.5cm}
\begin{tikzpicture}[
  scale=.75
]

\node
  at (.3,.55+.8)
  {
    \adjustbox{
      bgcolor=white,
      scale=.7
    }{
      \color{darkblue}
      \bf
      \def\arraystretch{.9}
      \begin{tabular}{c}
        surplus
        \\
        flux quantum:
        \\
        \color{purple}
        quasi-hole
        \\
        \color{purple}
        vortex
      \end{tabular}
    }
  };

\draw[
  dashed,
  fill=lightgray
]
  (0,0)
  -- (5,0)
  -- (7+.3-.1,2+.3)
  -- (2.8+.3+.1,2+.3)
  -- cycle;

\begin{scope}[
  shift={(2.4,.5)}
]
\shadedraw[
  draw opacity=0,
  inner color=olive,
  outer color=lightolive
]
  (0,0) ellipse (.7 and .3);
\end{scope}

\begin{scope}[
  shift={(4.5,1.5)}
]

\begin{scope}[
 scale=1.8
]
\shadedraw[
  draw opacity=0,
  inner color=olive,
  outer color=lightolive
]
  (0,0) ellipse (.7 and .25);
\end{scope}

\begin{scope}[
 scale=1.45
]
\shadedraw[
  draw opacity=0,
  inner color=olive,
  outer color=lightolive
]
  (0,0) ellipse (.7 and .25);
\end{scope}

\shadedraw[
  draw opacity=0,
  inner color=olive,
  outer color=lightolive
]
  (0,0) ellipse (.7 and .25);

\begin{scope}[
  scale=.2
]
\draw[
  fill=black
]
  (0,0) ellipse (.7 and .25);
\end{scope}

\end{scope}

\draw[
  white,
  line width=2
]
  (1.3, 1.8) .. controls 
  (2,2.2) and 
  (2.2,1.5) ..
  (2.32,.7);
\draw[
  -Latex,
  black!70
]
  (1.3, 1.8) .. controls 
  (2,2.2) and 
  (2.2,1.5) ..
  (2.32,.7);

\node
  at (1.3,2.7)
  {
    \adjustbox{
      scale=.7
    }{
      \color{darkblue}
      \bf
      \def\arraystretch{.9}
      \def\tabcolsep{-5pt}
      \begin{tabular}{c}
        $K$ flux-quanta
        absorbed
        \\
        by each electron:
      \end{tabular}
    }
  };

\draw[
 line width=2.5pt,
  white
]
  (2.4, 3.1) .. controls 
  (2.8,3.3) and 
  (4,3.5) ..
  (4.3,1.8);

\draw[
  -Latex,
  black!70
]
  (2.4, 3.1) .. controls 
  (2.8,3.3) and 
  (4,3.5) ..
  (4.3,1.8);

\node at 
  (5.9,2.6)
  {
   \scalebox{.8}{
     \color{gray}
     (cf. \cite[Fig. 16]{Stormer99})  
   }
  };

\node[
  gray,
  rotate=-20,
  scale=.73
] 
  at (4.8,+.3) {$S^2_{\!g}$};

\end{tikzpicture}
}
\hspace{-1cm}
  \adjustbox{raise=-1.4cm}{
  \begin{tikzpicture}[
    xscale=.7
  ]
    \draw[
      gray!30,
      fill=gray!30
    ]
      (-4.6,-1.5) --
      (+1.8,-1.5) --
      (+1.8+3-.5,-.4) --
      (-4.6+3+.5,-.4) -- cycle;

    \begin{scope}[
      shift={(-1,-1)},
      scale=1.2
    ]
    \shadedraw[
      draw opacity=0,
      inner color=olive,
      outer color=lightolive
    ]
      (0,0) ellipse (.7 and .1);
    \end{scope}

    \draw[
     line width=1.4
    ]
      (-1,-1) .. controls
      (-1,0) and
      (+1,0) ..
      (+1,+1);

  \begin{scope}
    \clip 
      (-1.5,-.2) rectangle (+1.5,1);
    \draw[
     line width=7,
     white
    ]
      (+1,-1) .. controls
      (+1,0) and
      (-1,0) ..
      (-1,+1);
  \end{scope}
  
    \begin{scope}[
      shift={(+1,-1)},
      scale=1.2
    ]
    \shadedraw[
      draw opacity=0,
      inner color=olive,
      outer color=lightolive
    ]
      (0,0) ellipse (.7 and .1);
    \end{scope}
    \draw[
     line width=1.4
    ]
      (+1,-1) .. controls
      (+1,0) and
      (-1,0) ..
      (-1,+1);

  \node[
    rotate=-25,
    scale=.7,
    gray
  ]
    at (3.2,-.58) {
      $S^2_{\!g}$
    };

  \draw[
    -Latex,
    gray
  ]
    (-3.4,-1.35) -- 
    node[
      near end, 
      sloped,
      scale=.7,
      yshift=7pt
      ] {time}
    (-3.4, 1.2);

  \node[
    scale=.7
  ] at 
    (0,-1.3)
   {\bf \color{darkblue} 
   surplus flux quanta};

  \node[
    scale=.7
  ] at 
    (1.5,0)
   {\bf \color{darkgreen} braiding};

  \node[
    fill=white,
    scale=.8
  ] at (-2.5,-1) {$
    \vert \Psi \rangle
  $};

  \node[
    fill=white,
    scale=.8
  ] at (-2.5,.7) {$
    \zeta
    \vert \Psi \rangle
  $};

  \draw[
    |->,
    black!80,
    line width=.5
  ]
    (-2.5,-.6) --
    (-2.5, .3);

  \end{tikzpicture}
  }
}
\end{figure}

This means that the interaction with the electron gas makes surplus magnetic flux in an FQH system \emph{effectively} behave differently than predicted by the ordinary Dirac charge quantization of \cref{OnOrdinaryMagneticFlux}. Therefore it stands to reason that the effective FQH flux is quantized (in the general sense of \emph{flux quantization} \cite{SS25-Flux}) instead by a deformation of the usual classifying space $B \mathrm{U}(1)$. But our \cref{Pi1OfMapsOutOfSurfaceIntoSphere} suggests that this deformation must be its 3-skeleton
\begin{equation}
  \begin{tikzcd}
    S^2 
    \simeq
    \mathbb{C}P^1
    \ar[
      r, 
      hook,
      "{ i }"
    ]
    &
    \mathbb{C}P^\infty
    \simeq
    B \mathrm{U}(1)
    \mathrlap{\,,}
  \end{tikzcd}
\end{equation}
because if we substitute that for $B \mathrm{U}(1)$ in the above computation \cref{OrdinaryMagneticFluxSectors}, then the experimentally observed braiding phase $\zeta$ does appear (by \cref{GeneralizedFinalTheorem}, \ref{GeneralizedEndResultFork=1}) as expected (cf. \parencites[(5.28)]{Tong2016}):
\begin{equation}
  \label{FQHSurplusFluxSectors}
  \begin{aligned}
    \text{FQH surplus flux sectors}
    &
    \defneq
    \pi_1\,
    \MapsComponent{n}\bracket({
      S^2_g
      ,
      S^2
    })
    \\
    & \simeq
    \mathrm{Heis}_{2g+1}(\mathbb{Z}; 2n)
    \mathrlap{\,.}
  \end{aligned}
\end{equation}
This match hence suggests the hypothesis (called ``Hypothesis h'' in \cite{SS25-FQH}) that 2-Cohomotopy is the correct global flux quantization law for effective anyonic FQH surplus flux quanta. 

Apart from neatly reproducing the quantum observables \cref{FQHSurplusFluxSectors} of solitonic FQH flux, this \emph{Hypothesis h} predicts experimentally relevant phenomena such as notably the possible attachment of nonabelian \emph{defect anyons} to superconducting islands inside the 2D electron gas (cf. \parencites{SS26-Islands}[Fig. D \& \S 3.8]{SS25-FQH}).

\subsubsection{FQH Anyon Braiding}
\label{OnFQHAnyonBraiding}

To see more concretely the actual braiding of anyons in this description, note that (as indicated in \cref{ThePontrjaginConstruction}) the \emph{Pontrjagin theorem} (cf. \parencites[\S II.16]{Bredon1993}[\S 3.2]{SS23-Mf}) identifies $\pi_0\, \Maps\bracket({S^2_g, S^2})$ with (cobordism classes of) normally framed codim=2 submanifolds of $S^2_g$, hence with \emph{signed points} in $S^2_g$, to be identified with the anyon cores seen in \cref{FQHAnyons}.

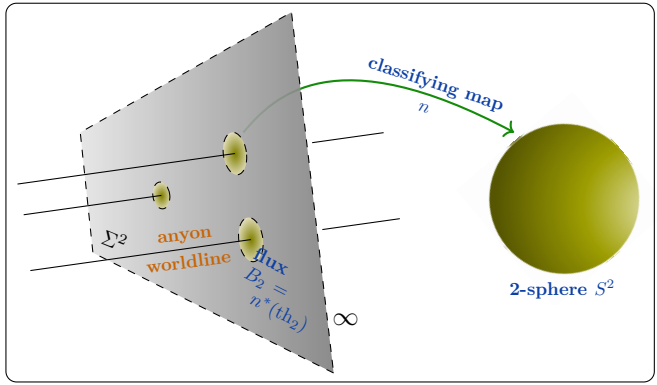
\begin{SCfigure}[1.2][htb]
\caption{
  \label{ThePontrjaginConstruction}
  The unstable \emph{Pontrjagin theorem} identifies homotopy classes of maps to the $n$-sphere with cobordism classes of normally framed codim=$n$ submanifolds. We may understand the latter as the cores of flux density quanta obtained as pullback of a \emph{Thom form} $\mathrm{th}_n$ of unit weight supported around $0 \in S^n$. 
}
\adjustbox{
  rndfbox=4pt
}{
\hspace{-.2cm}
\begin{tikzpicture}
\begin{scope}[
  scale=.8,
  shift={(.7,-4.9)}
]

\draw[
  line width=.8,
  ->,
  darkgreen
]
  (1.5,1) 
  .. controls (2,1.6) and (3,2.6) .. 
  (6,1);

\node[
  scale=.7,
  rotate=-18
] at (4.7,1.8) {
  \color{darkblue}
  \bf
  classifying map
};

\node[
  scale=.7,
  rotate=-18
] at (4.5,1.4) {
  \color{darkblue}
  \bf
  $n$
};

  \shade[
    right color=gray, left color=lightgray,
    fill opacity=.9
  ]
    (3,-3)
      --
    (-1,-1)
      --
        (-1.21,1)
      --
    (2.3,3);

  \draw[dashed]
    (3,-3)
      --
    (-1,-1)
      --
    (-1.21,1)
      --
    (2.3,3)
      --
    (3,-3);

  \node[
    scale=1
  ] at (3.2,-2.1)
  {$\infty$};

  \begin{scope}[rotate=(+8)]
  \shadedraw[
    dashed,
    inner color=olive,
    outer color=lightolive,
  ]
    (1.5,-1)
    ellipse
    (.2 and .37);
  \draw
   (1.5,-1)
   to 
    node[above, yshift=-1pt]{
     \;\;\;\;\;\;\;\;\;\;\;
     \rotatebox[origin=c]{7}{
     \scalebox{.7}{
     \color{darkorange}
     \bf
       anyon
     }
     }
   }
    node[below, yshift=+6.3pt]{
     \;\;\;\;\;\;\;\;\;\;\;\;
     \rotatebox[origin=c]{7}{
     \scalebox{.7}{
     \color{darkorange}
     \bf
       worldline
     }
     }
   }
   (-2.2,-1);
  \draw
   (1.5+1.2,-1)
   to
   (4,-1);
  \end{scope}

  \begin{scope}[shift={(-.2,1.4)}, scale=(.96)]
  \begin{scope}[rotate=(+8)]
  \shadedraw[
    dashed,
    inner color=olive,
    outer color=lightolive,
  ]
    (1.5,-1)
    ellipse
    (.2 and .37);
  \draw
   (1.5,-1)
   to
   (-2.3,-1);
  \draw
   (1.5+1.35,-1)
   to
   (4.1,-1);
  \end{scope}
  \end{scope}
  \begin{scope}[shift={(-1,.5)}, scale=(.7)]
  \begin{scope}[rotate=(+8)]
  \shadedraw[
    dashed,
    inner color=olive,
    outer color=lightolive,
  ]
    (1.5,-1)
    ellipse
    (.2 and .32);
  \draw
   (1.5,-1)
   to
   (-1.8,-1);
  \end{scope}
  \end{scope}
  
\end{scope}

\node[
  scale=.73,
  rotate=-27
] at (2.21,-5.21) {
  \color{darkblue}
  \bf
  \def\arraystretch{.9}
  \begin{tabular}{l}
    flux
    \\
    $B_2 =$ 
    \\
    $\;\;n^\ast(\mathrm{th}_2)$
  \end{tabular}
};

\node[
  scale=.73,
] at (2.05,-5) {
  \color{darkblue}
  \bf
};

\node[
  rotate=-140
] at (6,-4) {
  \includegraphics[width=2cm]{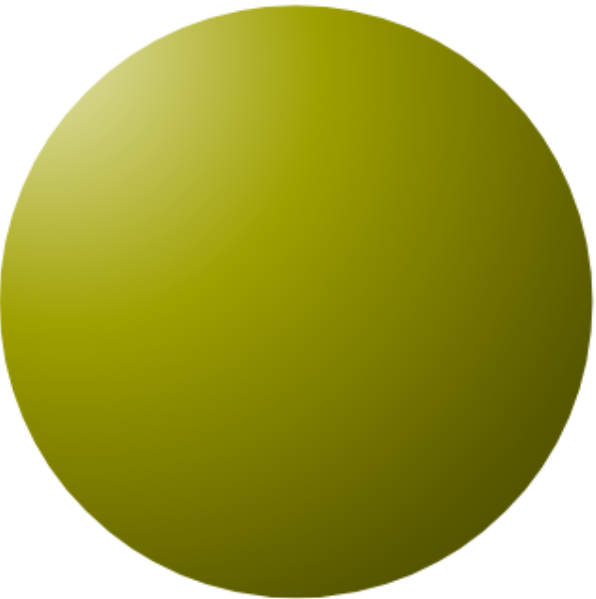}
};

\node[
  scale=.7
] at (6,-5.2) {
  \color{darkblue}
  \bf
  2-sphere $S^2$
};

\node[
  scale=.8,
  rotate=-22
] at (.05,-4.55)
{$\Sigma^2$};

\end{tikzpicture}
\hspace{-.7cm}
}
\end{SCfigure}

Along these lines, careful analysis shows \cite{SS25-AbelianAnyons} that the loop space $\Omega \PointedMaps\bracket({\mathbb{R}^2_{\cpt}, S^2})$ may be identified with the space of \emph{framed links} $L$ (of anyon worldlines) subject to framed link cobordism, and that under this identification the topological equivalence classes of these loops/worldlines are identified with the \emph{writhe} $\#L$ of the corresponding links, being the \emph{total crossing number} of any of their link diagrams: 
\begin{equation}
  \label{LinksToTheirWrithe}
  \begin{tikzcd}[row sep=0pt]
    \Omega\,
    \PointedMaps\bracket({
      \mathbb{R}^2_{\cpt}
      , 
      S^2
    })
    \ar[
      r,
      ->>
    ]
    &
    \pi_1
    \PointedMaps\bracket({
      \mathbb{R}^2_{\cpt}
      , 
      S^2
    })
    \simeq
    \mathbb{Z}
    \\
    L 
    \ar[
      r,
      |->,
      shorten=8pt
    ]
    &
    \#L
    \mathrlap{\,.}
  \end{tikzcd}
\end{equation}
But this says exactly that one power of the \emph{braiding phase} generator $\zeta$ is picked up for every crossing of anyon worldlines as seen in \cref{FQHAnyons} (cf. \cite[\S 3]{SS25-AbelianAnyons}).

\subsubsection{Identifying FQAH Anyons}

While FQH systems (\cref{OnSurplusFQHFlux,OnFQHAnyonBraiding}) thereby constitute the first (and currently only) experimentally verified candidate platform for genuine topological quantum hardware, the extremely low temperatures and strong magnetic fields they require obstruct their practical utility as such. It is therefore remarkable that, very recently, an ``anomalous'' version of fractional quantum Hall systems (FQAH, predicted by \cite{Tang2011,Sun2011,neupert2011fractional}, further developed in \cite{parameswaran2013fractional,roy2014band}, recently reviewed in \cite{ju2024fractional,moralesduran2024fractionalized,zhao2025exploring}) has been experimentally realized in various materials (\cite{Cai2023,Zeng2023,Park2023,Lu2024}): 

In these crystalline FQAH systems the role of the magnetic field in ``real space'' is played instead by an intrinsic property of the ``momentum space'' of the crystal electrons, called the \emph{Berry curvature} (cf. \parencites[\S 2]{Stanescu2020}[Fig. 3]{SS25-FQAH}).

\addedinline{
Concretely, the hallmark properties of anomalous Hall systems are (cf. \parencites{Sinitsyn2007,Nagaosa2010}[\S III.D]{Xiao2010}[\S 5.1]{Vanderbilt2018}) due to the \emph{Berry curvature}
$\Omega(\vec k)$ over the crystal's Brillouin torus of Bloch momenta, $\vec k \in \widehat{T}{}^2$, playing the role of magnetic flux density $F(\vec x)$ in the ordinary space of positions $\vec x$ inside FQH systems. 
}

\addedinline{
This is a veritable \emph{duality} (cf. Fig. \ref{PositionMomentumDuality}): The \emph{Lorentz force} law which underlies the ordinary Hall effect and the \emph{anomalous velocity law} which underlies the anomalous Hall effect are mirror images of each other under exchanging position space with momentum space and magnetic flux density with Berry curvature (an observation that goes back to \parencites{Chang1995}{Chang1996}[\S III.B]{SundaramNiu1999}, 
review in \parencites{Xiao2005}[(69)]{Sinitsyn2007}[\S III.A]{Xiao2010}{Chong2010}[(12-13)]{Stephanov2012}[(9-10)]{Culcer2024}). This exhibits the anomalous Hall effect as a Hall effect in momentum space.
}
\footnote{
  More in detail, units of Berry flux in momentum space are dual to quanta of \emph{surplus} magnetic flux (quasi-holes) in position space quantum Hall systems, cf. the careful analysis in \cite[\S A.1.2]{SS25-FQH}.
}

\begin{SCfigure}[.8][htb]
\caption{
  \label{PositionMomentumDuality}
  \addedinline{
  Under exchange of electron position with momentum, Lorentz force with anomalous velocity and magnetic flux with Berry curvature, the Hall effect is dual to its anomalous version (cf. \parencites[(9--11)]{Culcer2024}).
  }
}
\adjustbox{
  scale=0.9,
  rndfbox=4pt
}{
\begin{tikzpicture}

\node[scale=.7] at (-2.9,1.2) {%
  \color{gray}%
  \bf%
  \def\arraystretch{.9}%
  \begin{tabular}{@{}c@{}}%
    Semiclassical 
    \\
    equation of motion
    \\
    of crystal electrons
  \end{tabular}
};

\node at (0,0) {
\begin{tikzcd}[
  column sep=-2pt,
  row sep=10pt
]
  \dot {\vec k}
  &=&
  -
  \partial \epsilon
  /
  \partial\vec x
  &+&
  \overset{
    \mathclap{
      \adjustbox{
        scale=.7,
        raise=2pt
      }{
        \color{darkblue}
        \bf
        Lorentz force
      }
    }
  }{
    \vec x \,\times\, e\vec B
  }
  &\phantom{---}&
  \overset{
    \mathclap{
      \adjustbox{
        scale=.7,
        raise=4pt
      }{
        \color{darkblue}
        \bf
        \def\arraystretch{.9}
        \begin{tabular}{c}
          position
        \end{tabular}
      }
    }
  }{
    \vec x
  }
  \ar[
    <->,
    shorten=-2pt,
    dashed,
    gray,
    d
  ]
  &\phantom{---}&
  \overset{%
    {%
      \adjustbox{
        scale=.7,
      }{%
        \color{darkblue}%
        \bf%
        \def\arraystretch{.9}%
        \begin{tabular}{@{}c@{}}%
          magnetic
          \\
          flux density
        \end{tabular}%
      }%
    }%
  }{
    e \vec B
  }
  \ar[
    <->,
    shorten=-2pt,
    dashed,
    gray,
    d
  ]
  \\
  \dot {\vec x}
  &=&
  \phantom{+}
  \partial \epsilon
  /
  \partial\vec k
  &-&
  \underset{
    \mathclap{
      \adjustbox{
        scale=.7,
        raise=-2pt
      }{
        \color{darkblue}
        \bf
        anomalous velocity
      }
    }  
  }{
    \vec k \,\times\, \phantom{e}\vec\Omega
  }
  &&
  \underset{
    \mathclap{
      \adjustbox{
        scale=.7,
        raise=-3pt
      }{
        \color{darkblue}
        \bf
        \def\arraystretch{.9}
        \begin{tabular}{c}
          momentum
        \end{tabular}
      }
    }
  }{
    \vec k
  }
  &&
  \underset{
    \mathclap{
      \adjustbox{
        scale=.7,
        raise=-6.5pt
      }{
        \color{darkblue}
        \bf
        \def\arraystretch{.85}
        \begin{tabular}{c}
          Berry 
          \\
          curvature
        \end{tabular}
      }
    }
  }{
    \vec \Omega
  }
\end{tikzcd}
};
\begin{scope}[
  shift={(-2.3,0)}
]
\draw[
  <->,
  densely dashed
]
  (2.25,.5) .. controls
  (3.2,.5) and
  (3.2,-.5) ..
  (2.25,-.5);
\end{scope}

\node[scale=.8] at (2.6,1.5) {
  \color{darkorange}
  \bf
  Hall effect
};
\node[scale=.8] at (2.8,-1.5) {
  \color{darkorange}
  \bf
  \llap{Anomalous }Hall effect
};

\end{tikzpicture}
}
\end{SCfigure}

Therefore, if anyonic states in FQAH materials could be identified and controlled, this could open the door to \emph{practically viable} topological quantum hardware \addedinline{operating without the need for huge external magnetic fields and at less extreme temperatures.}

\addedinline{
  While this is strong motivation for developing a theory of FQAH anyons, we highlight that \emph{to date, no signatures of anyons have been experimentally detected in FQAH materials}.
}

\addedinline{
(For context: Even in ordinary FQH materials, experimentally explored since 1982, the first signature of anyons was found only in 2020 \cite{Nakamura2020,Bartolomei2020}, almost 40 years later. In comparison, FQAH materials were experimentally realized only a couple of years ago, in 2023.) 
}

\addedinline{
In this sense, our discussion of these anyons is \emph{speculative}, in that we take the mathematical result proven here and deduce from it --- under a plausible map of its content to the physical world that we recall now --- a \emph{prediction}. This prediction may or may not eventually be confirmed by experiment. Either way, this will be interesting for the mathematical physicist if the logic of the prediction is systematically developed --- which is what we briefly recall now.
}

While traditional theory has arguably remained inconclusive in identifying the nature and signature of potential FQAH anyons,
at this point the discussion in \cref{OnTheTheorem} applies, since: 
\begin{enumerate}
\item
Due to the translational symmetry of crystal lattices, the crystal momentum space is always a torus --- known as the \emph{Brillouin torus} $\widehat{T}^d$ (cf. \cite[\S 2.1]{Thiang2025}), which is 2-dimensional for the effectively 2-dimensional FQAH systems: $\widehat{T}^2 \simeq T^2$.

\item The most prominent FQAH systems have 2 relevant \emph{electron bands} 
\addedinline{
(cf. \parencites[(1)]{Chang2023}{GrushinNeupertChamonMudry2012}[\S 3]{IrkhinSkryabinYu2021}).
} 
This means that the space of choices of \emph{valence electron states} at each momentum $[k] \in \widehat{T}^2$ is the space of 1D complex subspaces among the 2D space of valence and conduction electron states, hence is the Grassmannian space $\mathrm{Gr}_1(\mathbb{C}^2) \simeq \mathbb{C}P^1 \simeq S^2$ (cf. \parencites[(8.3-4)]{Sergeev2023}[(4)]{SS25-FQAH}[Lem. 4.1]{SS25-Orient}).
\end{enumerate}
But this means that the moduli space of these crystal's topological parameters is just the kind of mapping space considered in \cref{OnTheTheorem}:
\begin{equation}
  \label{FCIParameterSpace}
  \text{Crystal parameter space}
  \simeq
  \Maps\bracket({
    \widehat{T}^2
    ,
    \mathbb{C}P^1
  })
  =
  \Maps\bracket({
    T^2
    ,
    S^2
  })
  \mathrlap{\,.}
\end{equation}

In particular, the \emph{topological phases} available to the FQAH system constitute the connected components of this space, which is the set of integers known as the \emph{Chern numbers} of the valence bundles (cf. \parencites[(8)]{SS25-FQAH}):
$$
  \begin{tikzcd}
    \pi_0\,
    \Maps\bracket({
      T^2
      ,
      S^2
    })
    \ar[
      r,
      "{\sim}",
      "{ i_\ast }"{swap}
    ]
    &
    \pi_0\,
    \Maps\bracket({
      T^2
      ,
      B \mathrm{U}(1)
    })
    \simeq
    \mathbb{Z}
    \mathrlap{\,,}
  \end{tikzcd}
$$
whence one refers to these topological quantum phases as (fractional) \emph{Chern insulators} (cf. \parencites[\S 8]{Sergeev2023}{neupert2011fractional}).

Furthermore, with \parencites{SS25-FQAH}{SS25-Crys} we may observe now that the \emph{quantum adiabatic theorem} entails (cf. \cref{AdiabaticQuantumTransport}) that anyonic topological order of quantum materials is classified by \emph{local systems} of Hilbert spaces (of gapped ground states) over the system's parameter space, hence here by linear representations of the fundamental groups of the parameter space \cref{FCIParameterSpace}:
\begin{equation}
  \substack{
    \text{Category of anyonic topological orders}
    \\
    \text{in topol. phase with Chern number $C$}
  }
  =
  \bracket({
    \pi_1\, 
    \MapsComponent{C}\bracket({
      T^2,
      S^2  
    })
  })\mathrm{Rep}\,.
\end{equation}
But this is exactly what \cref{FinalTheorem} \cref{EndResultFork=1} applies to, where it says that these anyonic quantum states are representations of the integer Heisenberg group --- just as expected for FQH systems on a torus.

\begin{figure}[htb]
\centering
\caption{
  \label{AdiabaticQuantumTransport}
 The adiabatic tuning of classical parameters $p$ along paths $\gamma$ in parameter space induces unitary transformations $U_\gamma$ between corresponding Hilbert spaces $\HilbertSpace$ of gapped ground states. For topological states these transformations depend only on the homotopy class of $\gamma$, exhibiting a \emph{local system} or \emph{flat bundle} of Hilbert spaces over the parameter space. These are equivalently linear representations of the fundamental groups of parameter loops $\ell$ at each base point, reflecting the \emph{topological order} of the system in any topological phase.
}
\centering
\adjustbox{
  rndfbox=4pt
}{
$
\begin{tikzcd}[
  decoration=snake,
]
  &
  \HilbertSpace_{p_2}
  \ar[
    dr,
    shorten=-2pt,
    "{ U_{\gamma_{{}_{23}}} }"{sloped}
  ]
  \\[-11pt]
  \HilbertSpace_{\mathrlap{p_{{}_{1}}}}
  \ar[
    rr,
    shorten >=-1pt,
    black,
    "{ U_{\gamma_{{}_{13}}} }"
  ]
  \ar[
    ur,
    shorten >=-2pt,
    "{
      U_{\gamma_{{}_{12}}}
    }"{sloped}
  ]
  &&
  \HilbertSpace_{p_{{}_3}}
  &
  \HilbertSpace_{p_{{}_0}}
  \ar[
    in=45+90,
    out=-45+90,
    looseness=4.7,
    shift left=2pt,
    shorten <=-2pt,
    "{ U_\ell }"
  ]
  \\[-20pt]
  &
  p_2
  \ar[
    dr,
    decorate,
    shorten <=-2pt,
    shorten >=-4pt,
    "{ \gamma_{{}_{23}} }"{yshift=2pt, sloped}
  ]
  &
  \\[-11pt]
  p_1
  \ar[
    rr, 
    decorate,
    shift right=1pt,
    shorten <=-2pt,
    shorten >=-3pt,
    "{
      \gamma_{{}_{13}}
    }"{yshift=2pt}
  ]
  \ar[
    ur,
    decorate,
    shorten <=-2pt,
    shorten >=-2pt,
    "{ 
      \gamma_{{}_{12}} 
    }"{yshift=2pt, sloped}
  ]
  &&
  p_{{}_3}
  &
  p_{\mathrlap{{}_0}}
  \ar[
      in=52+90,
      out=-52+90,
      looseness=5,
      shift left=4pt,
      decorate,
      shorten <=-1pt,
      shorten >=-4pt,
      "{ \ell }"
  ]
\end{tikzcd}
$
}
\end{figure}
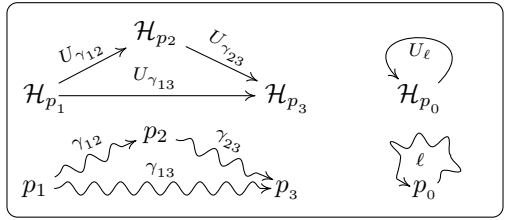

In conclusion, \cref{FinalTheorem} \cref{EndResultFork=1}, applied to the adiabatic monodromy (\cref{AdiabaticQuantumTransport}) in the topological parameter space \cref{FCIParameterSpace} of 2D 2-band (fractional) Chern insulators,
implies that the sought-after FQAH anyons are going to appear like the already observed anyons in FQH systems, but now \emph{localized in momentum space}. Further details on experimental signatures predicted by this result are discussed in \cite{SS25-Crys}.

\subsection{2-Dimensional Anyons via 4-Cohomotopy}
\label{On2DimensionalAnyonsVia4Cohomotopy}

In the above discussion of FQ(A)H anyons (\cref{OnOrdinaryFHQAnyonsIn2Cohomotopy}), the appearance of the 2-sphere and hence of 2-Cohomotopy as a quantization law is tied to the fact that both magnetic flux density and Berry curvature are differential 2-forms. For an analogous interpretation of the 4-sphere coefficient appearing in \cref{FinalTheorem} \cref{EndResultFork=2} and \cref{GeneralizedFinalTheorem} \cref{GeneralizedEndResultFork=2} one needs to understand 4-Cohomotopy as a flux quantization law of a kind of higher gauge field whose magnetic flux density is given by a differential 4-form.

\subsubsection{The C-field in 11D Supergravity}

A higher gauge field (the ``C-field'') with just such a higher 4-form flux density famously exists in 11-dimensional supergravity theory (11D SuGra, cf. \parencites{MiemiecSchnakenburg2006}{FreedmanVanProeyen:2012}[\S 3]{GSS24-SuGra}), where its electric Gau{\ss} law implies that 4-Cohomotopy is a valid choice (``Hypothesis H'', now with capital ``H'') for its flux quantization.

This means%
\footnote{
  Generally, if the tangent bundle of the spacetime does not trivialize, the full Hypothesis H demands (\parencites{FSS20-H}{FSS21-Hopf}{SS21-M5Anomaly}) that plain 4-Cohomotopy be replaced by \emph{tangentially twisted 4-Cohomotopy}, represented not by plain maps to the 4-sphere but by sections of 4-spherical fibrations over spacetime. But here we may disregard this further refinement, since the example spacetimes relevant here have trivializable tangent bundles, mainly owing to the fact that so does $S^3 \simeq \mathrm{SU}(2)$ and with it the product space $S^3 \times S^3$.
}
that when the gauge field content of 11D SuGra is globally completed according to Hypothesis H, then the solitonic topological observables on a globally hyperbolic spacetime $X^{1,10} \simeq \mathbb{R}^{1,1} \times X^{9}$  are spanned by
\begin{equation}
  \label{CFieldTopologicalPhaseSpace}
  \pi_0\,
  \PointedMaps\bracket({
    X^9_+
    \wedge
    \mathbb{R}^1_{\cpt}
    ,
    S^4
  })
  \simeq
  \pi_1\,
  \MapsComponent{0}\bracket({
    X^9
    ,
    S^4
  })  
  \mathrlap{\,.}
\end{equation}

\subsubsection{Near-horizons of intersecting M2/M5-branes }

In order that \cref{FinalTheorem} \cref{EndResultFork=2} applies, we need to look for solutions of 11D SuGra where the spacetime topology contains a factor topologically of the form $S^3 \times S^3$.

Such solutions do exist in the form (cf. \parencites[\S 2.2]{BoonstraPeetersSkendris1998}{GauntlettMyersTownsend1998})
\begin{equation}
  \label{TheAdS3Solution}
  X^{1,10}
  \simeq
  \mathrm{AdS}_3 
    \times 
  S^3 
    \times 
  S^3
    \times
  \mathbb{R}^2
\end{equation}
reflecting the near-horizon geometry of suitable M2/M5-brane intersections. (More generally, the $\sfrac{1}{2}$ BPS solutions of 11D SuGra of this form are products of $\mathrm{AdS}_3 \times S^3 \times S^3$ warped over a surface $\Sigma^2$, by \cite{DHokerEstesGutperleKryn2008a,DHokerEstesGutperleKryn2008b,DHokerEstesGutperleKryn2009}.)

\subsubsection{Anyonic monodromy in 11d SuGra}
\label{OnAnyonicMonodromyIn11DSuGra}

If we restrict attention to a bulk causal diamond $D^{1,2} \subset\mathrm{AdS}_3$ and take solitonic flux to be localized along one of the two factors of $\mathbb{R}^2$, the monodromy group \cref{CFieldTopologicalPhaseSpace} becomes
$$
  \begin{aligned}
  \pi_1\, \MapsComponent{0}\bracket({
    X^9
    ,
    S^4
  })
  & \simeq
  \pi_1\,
  \MapsComponent{0}\bracket({
    D^{1,2} \times \mathbb{R}^1
    \times 
    S^3 \times S^3
    ,\,
    S^4
  })
  \\
  &\simeq
  \pi_1\,
  \MapsComponent{0}\bracket({
    S^3 \times S^3
    ,\,
    S^4
  })
  \\
  & 
  \simeq
  \mathrm{Heis}_3\bracket({
    \mathbb{Z}, 0
  })
  \times^{\addedinline{\prime}}
  \CyclicGroup{12}
  \mathrlap{\,,}
  \end{aligned}
$$
where in the first step we used that $D^{1,2} \times \mathbb{R}^1$ is contractible, and in the last step we applied our new \cref{FinalTheorem} \cref{EndResultFork=2}.

But this means that, up to the extra torsion contribution, the algebra of topological observables on cohomotopically quantized C-field flux over such 11D SuGra backgrounds \cref{TheAdS3Solution} is exactly the same as for ordinary FQH anyons on the torus, according to \cref{FQHSurplusFluxSectors}, namely is the group algebra of $\mathrm{Heis}_3\bracket({\mathbb{Z},0})$. And with the algebra of observables also the corresponding topological quantum states are the same.

This is a novel realization of anyon statistics in quantum 11D SuGra. 
While the quantum observables are just those of the experimentally observed FQH anyons, these anyonic C-field flux quanta are, at face value, higher-dimensional objects, braiding in a higher-dimensional ambient space:%
\footnote{
 Alternatively, the usual point-like anyons in the cohomotopical guise of \cref{OnSurplusFQHFlux} may be ``geometrically engineered''  in 11D SuGra by realizing them on M5-probe worldvolumes (which in turn are embedded into 11D SuGra backgrounds, cf. \parencites{HoweSezgin1997}{GSS25-M5}), cf.  \parencites{SS25-Seifert}{SS25-WilsonLoops}, exposition in \parencites{SS25-Rickles}{SS25-Srni}.
}

\subsubsection{2-Dimensional anyons}

The analogue of the argument in \cref{OnFQHAnyonBraiding}, around \cref{ThePontrjaginConstruction}, via the Pontrjagin theorem shows that the anyonic C-field flux quanta identified above in \cref{OnAnyonicMonodromyIn11DSuGra} are 2-dimensional inside $S^3 \times S^3$ (hence are $2 + \mathrm{dim}(\mathrm{AdS}_3 \times \mathbb{R}^1) = 6$-dimensional inside $X^{1,10}$, hence are anyonic 5-branes of sorts). 
Hence, the analogue of \cref{LinksToTheirWrithe} says that the ``worldvolume'' traced out as these
higher-dimensional anyons propagate are 3-dimensional links inside $S^3 \times S^3$. 

Similar ``$d$-dimensional anyons'' have previously been discussed mostly for $d = 1$ as linear representations of \emph{loop braid groups} (cf. \cite{nLab:LoopBraidGroup}) potentially realizable in 3-dimensional quantum materials. The role of $d=2$-dimensional anyons in 6 ambient dimensions remains to be understood.

\appendix

\section{Some Topology}

For reference in the main text and to establish our notation, we briefly recall some notions and facts from basic topology.

We work in the category of compactly generated topological spaces, to be denoted $\mathrm{TopSp}$, and $\mathrm{TopSp}^\ast$ for the pointed version.
For $X, Y \in \mathrm{TopSp}$ we write
\begin{equation}
  \label{MappingSpace}
  \Maps\bracket({
    X, Y
  })
  \in 
  \mathrm{TopSp}
\end{equation}
for their \emph{mapping space}, and for $(X, x), (Y,y) \in \mathrm{TopSp}^\ast$ we write
\begin{equation}
  \label{PointedMappingSpace}
  \PointedMaps\bracket({
    X, Y
  })
  \subset
  \Maps\bracket({
    X, Y
  })
  \in 
  \mathrm{TopSp}
\end{equation}
for their \emph{pointed mapping space}.

For example, the circle is naturally a pointed space when realized as
$$
  S^1 
    \simeq 
  \frac{
    [0,1]
  }{
    \{0,1\}
  }
$$
and the \emph{based loop space} of a pointed space is
\begin{equation}
  \label{BasedLoopSpace}
  \Omega X
  :=
  \PointedMaps(S^1, X)
  \mathrlap{\,.}
\end{equation}

Given a pair of pointed spaces $(X,x), (Y,y) \in \mathrm{TopSp}^\ast$, their \emph{smash product} is
\begin{equation}
  \label{SmashProduct}
  X \wedge Y
  :=
  \frac{
    X \times Y
  }{
    X \!\times\! \{y\}
    \,\cup\,
    \{x\} \!\times\! Y
  }
  \mathrlap{\,.}
\end{equation}
Notably
\begin{equation}
  S^{n_1} \wedge S^{n_2}
  \simeq
  S^{n_1 + n_2},
\end{equation}
and generally 
\begin{equation}
  \label{Suspension}
  \Sigma X 
  :=
  S^1 \wedge X
\end{equation}
is the \emph{suspension} of a pointed space.

While the smash product is not Cartesian, it does have a natural diagonal map induced by the Cartesian diagonal, which we thus denote by the same symbol:
\begin{equation}
  \label{SmashDiagonal}
  \begin{tikzcd}[
    row sep=2pt, column sep=30pt
  ]
    X 
    \ar[
      rr,
      "{ \Delta }"
    ]
    \ar[
      dr,
      "{ 
        \Delta 
      }"{swap, sloped}
    ]
    &&
    X \wedge X
    \mathrlap{\,.}
    \\
    & 
    X \times X
    \ar[
      ur,
      ->>
    ]
  \end{tikzcd}
\end{equation}

The \emph{hom-adjunction} is the natural homeomorphism which regards a map of two arguments as a function of the first with values in maps of the second argument:
\begin{equation}
  \label{PointedHomAdjunction}
  \begin{tikzcd}
  \PointedMaps\bracket({
    X \wedge Y
    ,\,
    Z
  })
  \ar[
    rr,
    <->,
    "{ \widetilde{(-)} }",
    "{ \sim }"{swap}
  ]
  &&
  \PointedMaps\bracket({
    X
    ,\,
    \PointedMaps\bracket({
      Y, Z
    })
  })
  \mathrlap{.}
  \end{tikzcd}
\end{equation}
Notably the suspension operation \cref{Suspension} is left-adjoint to forming based loop spaces \cref{BasedLoopSpace}:
\begin{equation}
  \label{SuspensionLoopAdjunction}
  \begin{tikzcd}
    \PointedMaps\bracket({
      \Sigma X
      ,\,
      Y
    })
    \ar[
      r,
      "{ \widetilde{(-)} }",
      "{ \sim }"{swap}
    ]
    &
    \PointedMaps\bracket({
      X
      ,\,
      \Omega Y
    }).
  \end{tikzcd}
\end{equation}

Moreover, since the coproduct of pointed spaces is the \emph{wedge sum}
\begin{equation}
  X \vee Y
  :=
  \frac{ X \sqcup Y }{
    \{x,y\}
  }
  \mathrlap{.}
\end{equation}
we have
\begin{equation}
  \label{PointedMapsOutOfWedgeSum}
  \PointedMaps\bracket({
    X \vee Y
    ,\,
    Z
  })
  \simeq
  \PointedMaps\bracket({X, Z})
  \times
  \PointedMaps\bracket({Y, Z}).
\end{equation}

The \emph{homotopy groups} of $(X,x) \in \mathrm{TopSp}^\ast$ are
\begin{equation}
  \label{HomotopyGroupsViaMappingSpace}
  \pi_n(X,x)
  \simeq
  \pi_0
  \,
  \PointedMaps(S^n, X)
  \,.
\end{equation}

\section{Some Homotopy Theory}
\label{OnSomeHomotopyTheory}

For reference in the main text and to establish our notation, we briefly recall some notions and facts from basic homotopy theory.

A \emph{homotopy} is a path in a mapping space. The \emph{homotopy classes} of maps are the connected components of the mapping space \cref{PointedMappingSpace}:
\begin{definition}
\label[definition]{PointedHomotopyCategory}
The \emph{pointed homotopy category} $\mathrm{Ho}(\mathrm{TopSp}^\ast)$ has as objects the pointed topological spaces which have the homotopy type of CW-complexes, and as hom-sets the connected components of pointed mapping spaces between these:
$$
  \mathrm{Ho}(\mathrm{TopSp}^\ast)(X,Y)
  \simeq
  \pi_0 
  \,
  \mathrm{Map}^\ast(X,Y)
  \mathrlap{\,.}
$$
\end{definition}

In particular, the \emph{homotopy groups} of a pointed space $X$ are
\begin{equation}
  \label{HomotopyGroups}
  \pi_n(X)
  \defneq
  \pi_0\, \PointedMaps\bracket({
    S^n 
    ,\,
    X
  }).
\end{equation}
\begin{example}
Some well-known unstable homotopy groups of spheres in low degrees (cf. \cite[p. 186]{Toda1962}):
\begin{equation}
\label
{SomeUnstableHomotopyGroupsOfSpheres}
  \begin{aligned}
  \forall_{k < n \in \mathbb{N}}
  \;\;
  \pi_k(S^n) 
  &
  \simeq 
  \;\,0
  \mathrlap{\,,}
  \\
  \forall_{n \in \mathbb{N}_{\geq 1}}
  \;\;
  \pi_n(S^n) 
  &
  \simeq 
  \underbrace{\mathbb{Z}}_{
    \langle [\mathrm{id}] \rangle
  }
  \mathrlap{\,,}
  \end{aligned}
  \;\;\;\;
  \begin{aligned}
  \pi_3(S^2)
  & \simeq 
  \underbrace{
    \mathbb{Z} 
  }_{ \langle [h_{\mathbb{C}}] \rangle }
  \\
  \pi_7(S^4) 
  & \simeq 
  \underbrace{
    \mathbb{Z} 
  }_{ \langle [h_{\mathbb{H}}] \rangle }
  \times 
  \underbrace{
  \CyclicGroup{12}
  }_{
    \langle r_{\mathbb{H}} \rangle
  }
  \\
  \pi_{15}(S^8) 
  & \simeq 
  \underbrace{
    \mathbb{Z} 
  }_{ \langle [h_{\mathbb{O}}] \rangle }
    \times
 \underbrace{
   \CyclicGroup{120}
 }_{ \langle r_{\mathbb{O}} \rangle }
 \mathrlap{\,,}
  \end{aligned}
  \;\;\;\;
  \begin{aligned}
    \pi_6(S^4) 
    & \simeq \CyclicGroup{2}
    \\
    \pi_8(S^5)
    & \simeq \CyclicGroup{24}
    \\
    \pi_{14}(S^8) 
    & \simeq \CyclicGroup{2}
    \\
    \pi_{16}(S^9)
    & \simeq \CyclicGroup{240}
    \mathrlap{\,,}
  \end{aligned}
\end{equation}
where
$[h_{\mathbb{K}}]$ denotes the homotopy class of the $\mathbb{K} \in  \{\mathbb{C}, \mathbb{H}, \mathbb{O}\}$-Hopf fibration $h_{\mathbb{K}}$ (cf. \cite[\S 3.2.3]{SS25-Orient}), and $r_{\mathbb{K}}$ represents an unstable \emph{remainder} class.
\end{example}



Here the group structure on the homotopy groups is induced by the \emph{H-cogroup structure} on the spheres (cf. \cref{GroupStructureOnComponents} below):

\begin{definition}[{cf. \cite[\S 2.2]{Arkowitz2011}}]
\label[definition]{HGroup}
  An \emph{H-group} is a group internal to the pointed homotopy category (\cref{PointedHomotopyCategory}), hence a pointed topological space (with CW-structure) equipped with a binary operation and with inverses which satisfy the group axioms up to (unspecified) homotopy. 

  Dually, an \emph{H-cogroup} is a group object internal to the opposite of the pointed homotopy category.
\end{definition}
\begin{example}
\label[example]{LoopSpaceAsHGroup}
For $(X,x) \in \mathrm{TopSp}^\ast$:
\begin{enumerate}
 \item The loop space $\Omega X$ \cref{BasedLoopSpace} is an H-group (\cref{HGroup}) with binary operation given by concatenation of loops,
 \begin{equation}
   \label{LoopConcatenation}
   \begin{tikzcd}[
     ampersand replacement =\&,
    row sep=-3pt, column sep=0pt
   ]
     \Omega X \times \Omega X
     \ar[
       rr,
       "{ \star }"
      ]
     \&\&
     \Omega X
     \\
     (\ell_1, \ell_2)
     \&\longmapsto\&
     \bracket({
     s 
       \mapsto
     \begin{cases}
       \ell_1(2s) & \text{if } s \leq 1/2
       \\
       \ell_2(2s-1) & \text{if } s \geq 1/2
     \end{cases}
     })
     \mathrlap{,}
   \end{tikzcd}
 \end{equation}
 and inverses given by reversal of loops:
 \begin{equation}
   \begin{tikzcd}[row sep=-2pt, column sep=0pt]
     \Omega X 
     \ar[
       rr,
       "{ \overline{(-)} }"
     ]
     &&
     \Omega X
     \\
     \ell 
     &\longmapsto&
     \bracket({
       s 
         \mapsto 
       \ell(1-s)
     })
     \mathrlap{.}
   \end{tikzcd}
 \end{equation}
 
 \item The suspension $\Sigma X$ \cref{Suspension} is an H-cogroup with binary cooperation
 \begin{equation}
   \begin{tikzcd}[
     ampersand replacement =\&,
    row sep=-3pt, column sep=0pt
   ]
     \Sigma X
     \ar[rr]
     \&\&
     \Sigma X
     \vee
     \Sigma X
     \\
     (s,x)
     \&\longmapsto\&
     \begin{cases}
       \bracket({
         (2s,x), \ast
       })
       & \text{if } s \leq 1/2
       \\
       \bracket({
         \ast,
         ((2s-1),x)
       })
       & \text{if } s \geq 1/2\,,
     \end{cases}
   \end{tikzcd}
 \end{equation}
 and coinverses given by
 \begin{equation}
   \begin{tikzcd}[row sep=-2pt, column sep=0pt]
     \Sigma X
     \ar[
       rr
     ]
     &&
     \Sigma X
     \\
     (s,x) 
       &\longmapsto&
     (1-s, x)
     \mathrlap{\,.}
   \end{tikzcd}
 \end{equation}
\end{enumerate}
Here (co)associativity is witnessed by ``thin'' homotopies which exhibit reparameterizations of  concatenated loops.
\end{example}

\begin{example}
\label[example]{HGroupStructureOnMappingSpace}
  More generally, the mapping spaces \cref{SuspensionLoopAdjunction} out of a suspension or into a loop space
  inherit H-group structure by pointwise group operation, hence where the binary operation on a pair of maps
  $
   \inlinetikzcd{
     f,g
     :
     X \ar[r] \& \Omega A
   }
  $
  is given by the following composite 
  with the diagonal \cref{SmashDiagonal} on the left and loop concatenation \cref{LoopConcatenation} on the right:
  $$
    \begin{tikzcd}
      X
      \ar[
        d,
        "{ 
          \Delta 
        }"{swap,pos=.42}
      ]
      \ar[
        rr,
        "{ f \star g }"
      ]
      &&
      \Omega A
      \\
      X \times X
      \ar[
        rr,
        "{ f \times g }"
      ]
      &&
      \Omega A \times \Omega A
      \mathrlap{\,.}
      \ar[
        u,
        "{ \star }"
      ]
    \end{tikzcd}
  $$
\end{example}

\begin{remark}
\label[remark]{GroupStructureOnComponents}
  Under passage to homotopy classes of maps, H-(co)group structure induces genuine group structure: For $(X,x), (A,a) \in \mathrm{TopSp}^\ast$ the H-group structure on the mapping space (\cref{HGroupStructureOnMappingSpace}) induces actual group structures
  \begin{equation}
    \label{IncaranationsOfFundGrpdOfMappingSpace}
    \begin{aligned}
    \pi_0\,
    \PointedMaps\bracket({X, \Omega A})
    & \simeq
    \pi_0\, 
    \PointedMaps\bracket({\Sigma X, A})
    \\
    & \simeq
    \pi_1\, 
    \PointedMaps\bracket({X,A})
    \;\in\;
    \mathrm{Grp}
    \mathrlap{\,.}
    \end{aligned}
  \end{equation}
  In particular, the group structure on homotopy groups \cref{HomotopyGroups} in positive degree, 
  $$
    \pi_{d+1}(A)
    \simeq
    \pi_0\,
    \PointedMaps(
      \Sigma S^d
      ,\,
      A
    )
    \mathrlap{\,,}
  $$
  arises this way.  
\end{remark}

Further in this vein:
\begin{example}
\label[example]{PinchingHCoaction}
  For $X^d$ a CW-complex of dimension $d \geq 1$, and $\inlinetikzcd{ S^{d-1} \ar[r, hook ] \& D^d \ar[r, hook] \&  X}$ an embedding of the boundary of a $d$-ball into the interior of a $d$-cell of $X^d$, then the \emph{pinch map} $\phi$ given by the homotopy pushout
  \begin{equation}
    \label{PinchMapPushout}
    \begin{tikzcd}[row sep=12pt, column sep=45pt]
      S^{d-1}
      \ar[
        d,
        hook
      ]
      \ar[r]
      \ar[
        dr,
        phantom,
        "{ \ulcorner }"{pos=.9}
      ]
      &
      \ast
      \ar[d]
      \\
      X^d 
      \ar[
        r,
        "{ \phi }"
      ]
      &
      X^d \vee S^d
    \end{tikzcd}
  \end{equation}
  exhibits an ``H-coaction'' (a coaction up to unspecified homotopies) of the H-cogroup $S^d \simeq \Sigma S^{d-1}$ (from \cref{LoopSpaceAsHGroup}). For any $A \in \mathrm{TopSp}^\ast$ this entails an actual action of the homotopy group $\pi_d(A) \simeq \pi_0\PointedMaps(S^{d}, A)$ (from \cref{GroupStructureOnComponents}) on the homotopy classes of maps from $X^d$ to $A$:
  \begin{equation}
    \label{CoPinchAction}
    \begin{tikzcd}[column sep=16pt]
      \pi_d(A)
      \times
      \pi_0\,
      \PointedMaps(
        X^d
        ,\,
        A
      )
      \ar[
        r,
        "{ \sim }"
      ]
      &
      \pi_0\,
      \PointedMaps(
        X^d \vee S^d
        ,\,
        A
      )
      \ar[
        r,
        "{ \phi^\ast }"
      ]
      &
      \pi_0\,
      \PointedMaps(
        X^d
        ,\,
        A
      )
      \mathrlap{\,.}
    \end{tikzcd}
  \end{equation}
  This is used in the proof of \cref{ConnectedComponentsOfMappingSpaceAreEquiv} above. 
\end{example}


\printbibliography

\end{document}